\documentclass{article}
\pdfoutput=1
\usepackage{ijcai15}
\usepackage{times}
\usepackage{mathptmx}
\DeclareMathAlphabet{\mathcal}{OMS}{cmsy}{m}{n}
\DeclareSymbolFont{greekletters}{OML}{cmm}{m}{it}
\DeclareMathSymbol{\pi}{\mathalpha}{greekletters}{"19}
\DeclareMathSymbol{\varphi}{\mathalpha}{greekletters}{"27}

\usepackage{amssymb,amsmath,amsthm,mathtools}
\usepackage{microtype}
\usepackage{stmaryrd}

\usepackage{tikz}
\usepackage[all]{xy}
\usepackage{wrapfig}
\usepackage[font=small]{caption}
\usepackage{booktabs}

\newtheorem{theorem}{Theorem}
\newtheorem*{theorem-non}{Theorem}
\newtheorem{lemma}{Lemma}
\newtheoremstyle{moreresults} 
{\topsep}                    
{\topsep}                    
{}                   
{}                           
{\scshape}                   
{.}                          
{.4em}                       
{}  
\theoremstyle{moreresults}
\newtheorem{remark}{Remark}

\newcommand{\tox}[2]{\{\toxic{#1}{#2}\}}
\newcommand{\utilties}[1]{\llbracket#1\rrbracket}
\newcommand{\orderings}{(\ge_i)_{i\in N}}
\newcommand{\W}{\mathcal W}	

\newcommand{\inn}{\subseteq} 
\newcommand{\pref}{\succcurlyeq} 
\newcommand{\N}{\mathcal N} 
\newcommand{\C}{\mathcal C} 
\newcommand{\scC}{\scalebox{0.85}{$\C$}}
\newcommand{\A}{\mathcal A} 
\newcommand{\on}{\operatorname}
\renewcommand*{\le}{\leqslant}
\renewcommand*{\ge}{\geqslant}
\renewcommand{\emptyset}{\varnothing}

\newcommand{\NP}{NP}
\newcommand{\threesat}{\textsc{\scalebox{0.85}{3}sat}}

\newcommand{\toxic}[2]{\ensuremath{#1\text{-}#2}}
\newcommand{\toxicc}[2]{\ensuremath{\{#1\text{-}#2\}}}

\makeatletter
\newsavebox{\@brx}
\newcommand{\llangle}[1][]{\savebox{\@brx}{\(\m@th{#1\langle}\)}%
	\mathopen{\copy\@brx\kern-0.5\wd\@brx\usebox{\@brx}}}
\newcommand{\rrangle}[1][]{\savebox{\@brx}{\(\m@th{#1\rangle}\)}%
	\mathclose{\copy\@brx\kern-0.5\wd\@brx\usebox{\@brx}}}
\makeatother
\newcommand{\longto}{\longrightarrow}
\makeatletter
\newcommand{\oset}[3][0ex]{%
	\mathrel{\mathop{#3}\limits^{
			\vbox to#1{\kern-2\ex@
				\hbox{$\scriptstyle#2$}\vss}}}}
\makeatother

\newcommand{\SSNS}{\textsc{ssns}}
\newcommand{\SNS}{\textsc{sns}}
\newcommand{\SIS}{\textsc{sis}}
\newcommand{\SC}{\textsc{scr}}
\newcommand{\Core}{\textsc{cr}}
\newcommand{\IS}{\textsc{is}}
\newcommand{\NS}{\textsc{ns}}

\newcommand{\property}[1]{\vspace{4pt}\noindent\hangindent=0.25cm\textit{#1.}} 
\newcommand{\class}[1]{\vspace{3pt}\noindent\hangindent=0.2cm\textbf{#1.}} 

\title{Simple Causes of Complexity in Hedonic Games\thanks{Presented at IJCAI 2015, Buenos Aires. This version includes full proofs in the appendix.}} 
\author{Dominik Peters \and Edith Elkind \\ Department of Computer Science\\ University of Oxford, UK \\ $\{$dominik.peters, edith.elkind$\}$@cs.ox.ac.uk}

\begin{document}

\maketitle

\begin{abstract}
  Hedonic games provide a natural model of coalition formation among self-interested agents. The associated problem 
of finding stable outcomes in such games has been extensively studied. In this paper, we identify simple conditions on 
expressivity of hedonic games that are sufficient for 
the problem of checking whether a given game admits a stable outcome
to be computationally hard. 
Somewhat surprisingly, these conditions are very mild and intuitive.
Our results apply to 
a wide range of stability concepts (core stability, individual stability, Nash stability, etc.) and to many known 
formalisms for hedonic games (additively separable games, games with $\mathcal W$-preferences, fractional hedonic 
games, etc.), and unify and extend known results for these formalisms. They also have broader applicability: 
for several classes of hedonic games whose computational complexity has not been explored
in prior work,
we show that our framework immediately implies a number of hardness results for them.
\end{abstract}

\section{Introduction}\label{sec:intro}
Hedonic games \cite{Dreze1980,Banerjee2001,Bogomolnaia2002} 
provide an elegant and versatile model of coalition formation among
strategic agents. In such games, each agent has preferences over {\em coalitions}
(subsets of players) that she can be a part of, and an outcome of the game
is a partition of agents into coalitions. Clearly, the quality of an outcome
depends on how well it reflects the agents' preferences. In particular,
it is desirable to have outcomes that are {\em stable}, i.e., 
do not offer the agents an opportunity to profitably deviate. Many different concepts
of stability have been proposed in the hedonic games literature (see Section~\ref{sec:prelim}
for a brief summary, and \cite{azizsavani} for an in-depth discussion), 
and for each of them a natural computational 
question is whether a given game admits an outcome that is stable in that sense.

\newcommand{\knownP}{(P)}
\newcommand{\knownNPC}{NP-c.}
\newcommand{\knownNPH}{NP-h.}
\newcommand{\newNPC}{\colorbox{black!17}{NP-c.}}
\newcommand{\newNPH}{\colorbox{black!17}{NP-h.}}
\newcommand{\newNPCwasopen}{\colorbox{black!17}{NP-c.}}
\newcommand{\newNPHwasopen}{\colorbox{black!17}{NP-h.}}
\newcommand{\open}{}
\newcommand{\trivial}{(+)}
\renewcommand{\arraystretch}{1.3}	
\begin{table}[t]
	
	\centering
	
	{\footnotesize
		\setlength{\tabcolsep}{0.7pt}
\begin{tabular}{lccccc} 
	\toprule
	& {\large\sc sns} & {\large\sc scr} & {\large\sc cr} & {\large\sc ns} & {\large\sc is} \\ \midrule

IRCL of length $\le 3$ & \newNPH & \open & \knownNPC & \knownNPC & \knownNPC \\
	
IRCL of length $\le 9$ & \newNPH & \newNPC & \knownNPC & \knownNPC & \knownNPC \\

Hedonic Coalition Nets & \newNPH & \newNPH & \knownNPH & \newNPC & \newNPC \\

Stable Marriages (\textsc{SMI}) & \open & \open & \knownP & \newNPC & \knownP \\
	
$\mathcal{W}$-preferences (no ties) & \open & \knownP & \knownP & \knownNPC & \newNPC \\

$\mathcal{W}$-preferences & \newNPH & \open & \knownNPC & \knownNPC & \knownNPC \\

$\mathcal{WB}$-preferences (no ties)\:\:\:& \open & \knownP & \knownP & \newNPCwasopen & \newNPC \\

$\mathcal{WB}$-preferences & \newNPH & \open & \newNPCwasopen 
 & \newNPCwasopen & \newNPCwasopen \\

B- \& W-hedonic games & \newNPH & \open & \newNPH & \newNPC & \newNPC \\

Additively separable & \newNPH & \knownNPH & \knownNPH & \knownNPC & \knownNPC \\

Fractional hedonic games & \newNPH & \newNPH & \knownNPH & \knownNPC & \knownNPC \\

Social FHGs & \open & \newNPH & \newNPHwasopen & \trivial & \trivial \\

Median & \open & \newNPHwasopen & \newNPHwasopen & \open & \open \\

Midrange ($\raisebox{0.85pt}{$\scriptstyle\frac12$}\mathcal B + \raisebox{0.85pt}{$\scriptstyle\frac12$}\mathcal W$) & \newNPH & \open & \newNPH & \newNPC & \newNPC \\

3-Approval & \newNPH & \open & \newNPH & \newNPC & \newNPC \\

4-Approval & \newNPH & \newNPH & \newNPH & \newNPC & \newNPC \\

\bottomrule
\end{tabular}}
\caption{Some of the hardness results implied by our framework for the problem of identifying hedonic games with stable outcomes. 
Gray entries are results that have not appeared in the literature before. 
(P) indicates known polynomial-time algorithms, (+) means that a stable outcome always exists. See Section \ref{sec:applications} for details.}
\label{table:results}
\vspace{-6pt}
\end{table}

The complexity of this question depends on how the game is represented:
while every hedonic game can be described by explicitly listing each agent's preference
relation over all coalitions that may contain her, in recent years there 
has been a considerable amount of research on {\em succinct} representation
formalisms for hedonic games, i.e., ones where a game description
size scales polynomially 
with the number of agents~$n$. Typically, such formalisms
are not universally expressive, but capture important
classes of hedonic games. For instance, if the utility that an agent 
assigns to a coalition is given by the sum/average/minimum/maximum of the utilities 
she assigns to individual members of that coalition, the entire game
can be described by $n(n-1)$ numbers (such games are known as, respectively,
additively separable games \cite{Bogomolnaia2002}, fractional hedonic games \cite{Aziz2014a}, 
and games with 
${\cal W}$- and ${\cal B}$-preferences \cite{Hajdukova2003,Cechlarova2004}). There are also representation
formalisms that are universally expressive (and hence exponentially verbose
in the worst case), but provide succinct descriptions of hedonic games 
that have certain structural properties; examples include Individually Rational 
Coalition Lists \cite{Ballester2004} and Hedonic Coalition Nets \cite{Elkind2009}. 
The complexity of stability-related problems under these and other representations
for hedonic games has been investigated by a number of researchers 
(see \cite{WoegingerSurvey} for a survey); with a few exceptions, checking whether
a game admits a stable outcome turns out to be computationally hard.

In this paper, we unify and extend several known hardness results for this family of problems 
in order to uncover common causes of complexity of stability-related questions in hedonic games. 
In their simplest form, our results imply that if in a given representation formalism, agents are 
able to rank coalitions of size two in any way they wish, and if agents are 
to some extent averse to the presence of enemies, then the problem of checking whether a game admits a stable 
outcome is NP-hard. The precise meaning of being averse to enemies depends on the stability concept in question. We 
also introduce intuitively appealing conditions on how agents rank coalitions of size three, which turn 
out to entail NP-hardness even if the underlying preferences are strict.
Our approach enables us to automatically derive
new hardness results for hedonic games: instead of coming up with a hardness
reduction, one can simply check whether the representation in question
satisfies the relevant conditions on enemy-aversion and coalitions of size two or three. 
By doing so, we answer several questions that
were left open by prior work, and substantially contribute to the understanding
of computational complexity of 
somewhat less explored 
solution concepts: to the best of our knowledge, we are the first to obtain NP-hardness results for strong
Nash stability ({\sc sns}), strict strong Nash stability ({\sc ssns}), 
and strong individual stability ({\sc sis}). 

To provide further evidence of the power of our approach,
we also consider several classes of hedonic games whose complexity
has not been investigated before, and derive NP-hardness results for them using
our methodology. Perhaps the most interesting of them is the class of 
{\em median games}, proposed by \cite{Hajdukova2006}, where each agent assigns a utility to every other agent,
and her utility for a coalition is the utility she assigns to the median agent in that 
coalition. 

The complexity results implied by our analysis are summarized 
in Table~\ref{table:results}. However, we believe that the sufficient conditions
for hardness identified in our work are at least as important as the specific
new results we have established. Indeed, these conditions indicate which
additional constraints should be placed on a representation formalism to avoid
the complexity trap,
and may guide researchers towards identifying formalisms that adequately describe their 
application scenario, yet admit efficient algorithms for finding stable outcomes.

\section{Preliminaries}\label{sec:prelim}
Given a finite set of agents $N = \{1,\dots,n\}$, a \textit{hedonic game} is a pair 
$G = \langle N, (\pref_i)_{i\in N}\rangle$, where 
$\pref_i$ is a complete and transitive preference relation over $\N_i = \{ S \inn N : i \in S \}$. 
We write $S\succ_i T$ when $S\pref_i T$, but $T\not\pref_i S$.
A \textit{class $\C$ of hedonic games} is any collection of hedonic games. 
We say that a class $\C$ is \textit{polynomially representable} 
if there exists a polynomial $p(x)$ and a poly-time algorithm $A$ such that 
each $\langle N, (\pref_i)_{i\in N}\rangle\in\C$ can be represented by a binary string of length at most $p(|N|)$,
and, given this string, an agent $i\in N$, and a pair of coalitions $S,T\in \N_i$, 
algorithm $A$ can decide whether $S\pref_i T$. For example, additively separable hedonic games 
mentioned in Section~1 form a polynomially representable class.

An {\em outcome} of a hedonic game is a partition $\pi$ of $N$ into disjoint coalitions. We write $\pi(i)$ for the 
coalition of $\pi$ that contains $i$. For partitions $\pi$ and $\pi'$, we write $\pi \pref_i \pi'$ to mean  
$\pi(i) \pref_i \pi'(i)$. 

We are mainly interested in the \textit{stability} of a given partition $\pi$ of $N$. 
We will consider seven stability concepts for hedonic games: two that are based on 
individual deviations, and five that are based on group deviations. The former group comprises \textit{Nash stability} 
(\NS) and \textit{individual stability} (\IS). A partition $\pi$ is {\NS} if no player can benefit from moving to 
another (possibly empty) coalition $S$ in $\pi$, i.e., $\pi(i)\pref_i S\cup\{i\}$ for all $S\in\pi\cup\{\emptyset\}$.
Partition $\pi$ satisfies {\IS} if no player can make such a beneficial move without making an 
agent in $S$ worse off, i.e., for each $S\in\pi\cup\{\emptyset\}$ it holds that 
$\pi(i)\pref_i S\cup\{i\}$ or $S\succ_j S\cup\{i\}$ for some $j\in S$.

The classic solution concept for group deviations is the \textit{core} (\Core). We say that a non-empty 
coalition $S$ {\em \Core-blocks} $\pi$ if $S\succ_i \pi(i)$ for all $i\in S$; 
it {\SC-blocks}  
$\pi$ if $S\pref_i \pi(i)$ for all $i\in S$ and, moreover, $S\succ_i \pi(i)$ for some $i\in S$.
If no coalition \Core-blocks $\pi$, it is in the {\em core} (\Core); if no coalition \SC-blocks it,
it is in the {\em strict core} (\SC).

\setlength\intextsep{0pt}
\begin{wrapfigure}[9]{r}{0pt}
        \hspace{-12pt}
        \scalebox{1}{\begin{tikzpicture}[scale=0.48]
		\tikzstyle{pfeil}=[->,draw]
		\tikzstyle{onlytext}=[font=\scshape, inner sep=2.3pt]
		
		\node[onlytext] (SSNS) at (4,4.5) {ssns};
		\node[onlytext] (SNS) at (2,3) {sns};
		\node[onlytext] (NS) at (2,1.5) {ns};
		\node[onlytext] (SIS) at (4,1.5) {sis};
		\node[onlytext] (SCR) at (6,3) {scr};
		\node[onlytext] (IS) at (2,0) {is};
		\node[onlytext] (IR) at (4,-1.5) {ir};
		\node[onlytext] (CR) at (6,0) {cr};
		
		\draw[pfeil] (IS) to (IR);
		\draw[pfeil] (CR) to (IR);
		\draw[pfeil] (SSNS) to (SCR);
		\draw[pfeil] (SSNS) to (SNS);
		\draw[pfeil] (SNS) to (NS);
		\draw[pfeil] (SIS) to (IS);
		\draw[pfeil] (SCR) to (SIS);
		\draw[pfeil] (SIS) to (CR);
		\draw[pfeil] (SNS) to (SIS);
		\draw[pfeil] (NS) to (IS);
		\draw[pfeil] (SCR) to (CR);
\end{tikzpicture}}
\end{wrapfigure}

Karakaya~\shortcite{Karakaya2011} introduced \textit{strong Nash stability} (\SNS), 
and Aziz and Brandl~\shortcite{Aziz2012} introduced the derived 
notions of \textit{strict strong Nash stability} (\SSNS) and \textit{strong individual stability} (\SIS). 
These solution concepts deal with deviations where 
the deviators do not necessarily form a single coalition.
Given two partitions $\pi$, $\pi'$, 
we say that a coalition $H\inn N$ can {\em reach} $\pi'$ from $\pi$ 
if for all 
$i,j\not\in H$ we have $\pi(i) = \pi(j)$ if and only if 
$\pi'(i) = \pi'(j)$. 
Coalition $H$ {\em \SSNS-blocks} $\pi$ if it can reach 
some $\pi'$ with $\pi'\pref_i\pi$ for all $i\in H$ and $\pi'\succ_i\pi$ for some $i\in H$. 
If $H$ can reach some $\pi'$ with $\pi'\succ_i\pi$ for all $i\in H$
then $H$ is said to {\em \SNS-block} $\pi$. If $H$ \SNS-blocks $\pi$ by reaching $\pi'$ and, moreover,
for each $i\in H$ and each $j\in \pi'(i)$ we have $\pi' \pref_j \pi$, 
then $H$ is said to {\em \SIS-block}~$\pi$. A partition 
$\pi$ is $\alpha$-stable (where $\alpha\in\{\SSNS,\SNS,\SIS\}$) if no coalition $\alpha$-blocks it.
Intuitively, \SNS-blocking coalitions allow groups of agents to swap places with each other. For \SIS-blocking 
coalitions, agents joined by a deviator must consent to the changes.

The diagram above shows implication relationships among these concepts. A partition that is \SSNS-stable is also 
stable under every other solution concept considered here. A coalition $S\ni i$ is \textit{individually 
rational} ({\sc ir}) for $i$ if $S\pref_i \{i\}$. A partition $\pi$ is said to be {\sc ir} if 
$\pi(i)$ is {\sc ir} for all $i\in N$.

\section{Properties of Preferences}\label{sec:properties}
Our hardness results require a given class $\C$ of hedonic games to be expressive enough to include hard instances. 
To this end, agents should have some freedom in how they order small coalitions. 
Our results apply
to formalisms that enable each agent $i$ to express arbitrary 
preferences
over coalitions of the form $\{i,j\}$, as well as satisfy a few other constraints.
Thus, in a sense, our results are about the 
hardness of finding stable outcomes in hedonic games obtained 
by \emph{lifting} preferences over individual players to preferences over coalitions.

We associate each agent $i\in N$ with a complete and transitive order $\ge_i$ over 
$N$. We interpret $\ge_i$ as $i$'s preference order over the set of players. 
We write $j >_i k$ if $j\ge_i k$ but not $k\ge_i j$, and we write
$j\sim_i k$ if both $j\ge_i k$ and $k\ge_i j$.
We call $F_i=\{j\neq i : j\ge_i i \}$ and $E_i=\{j\neq i : j<_i i \}$
the {\em friends} and the {\em enemies} of~$i$. 
In what follows, it will not matter how $\ge_i$ 
orders $E_i$---only its restriction on $F_i$ will be of interest.

We now describe a series of properties that 
relate $i$'s preferences $\pref_i$ over the coalitions in $\N_i$  
to her preferences $\ge_i$ over the agents in $N$. These properties express various ways in which $\pref_i$ can be said 
to \textit{extend} $\ge_i$. The numerical examples in brackets aim to illustrate the intuition behind
these properties.

\vspace{3pt}
\property{Consistent on pairs}
For all $j,k \in F_i \cup \{i\}$
it holds that $\{i,j\}\pref_i \{i,k\}$ iff $j\ge_i k$.

\property{Monotone on triangles {\small (`$\mathit{7+6>7+5}$')}}
If $j,j',k,k'\in F_i$ are such that $j \ge_i j' >_i k >_i k'$, then $\{i,j,k\}\succ_i \{i,j',k'\}$.

\property{Triangle-appreciating  {\small (`$\mathit{7+5>7}$')}}
Two almost equally good friends together are preferable to the better friend alone: 
If $j,k,\ell\in F_i$ are ranked $j >_i k >_i \ell$ and they are immediate successors under $>_i$, 
then $\{i,j,\ell\} \succ_i \{i,j\}$.

\vspace{5pt}
\noindent 

Only few polynomial-time algorithms for finding stable outcomes in hedonic games are known, mainly 
confined to matching problems and the (structurally similar) $\mathcal W$-hedonic games. Notably these 
classes of games fail to 
be triangle-appreciating, and in view of our results in Section~\ref{sec:claims} 
this is a key reason why they admit easiness results.

The following properties express that agents do not like coalitions that contain too many enemies.
\vspace{3pt}

\property{$\{\toxic{a}{b}\}$-toxic} 
If $|S\cap F_i|=a$, but $|S\cap E_i|\ge b$ then $\{i\} \pref_i S$.

\property{Strictly $\{\toxic{a}{b}\}$-toxic} 
Same as above with $\{i\} \succ_i S$.

\property{Weakly $\{\toxic{a}{b}\}$-toxic}
Same as above with $\{i,j\} \succ_i S$ for all $j\in F_i$.

\property{Intolerant in triangles}
If $E_i'\inn E_i$ is non-empty and $j,k\in F_i$ are distinct then $\{i,j,k\}\succ_i \{i,j,k\}\cup E_i'$.

\vspace{5pt}
\noindent We write `(strictly/weakly) $\{ \toxic{a_1}{b_1},\dots, \toxic{a_m}{b_m} \}$-toxic' for preferences that are (strictly/weakly) $\{\toxic{a_t}{b_t}\}$-toxic for $t\hspace{-0.1pt}=1,\dots,m$.

Given a collection $(\ge_i)_{i\in N}$ of orderings, we say that a hedonic game $\langle N, (\pref_i)_{i\in N}\rangle$ 
satisfies one of the properties above if each $\pref_i$ satisfies it with respect to $\ge_i$. 
We say that the collection 
is \textit{strict} if each $\ge_i$ is antisymmetric, so $j\neq k$ implies $j\not\sim_i k$.
The collection is \textit{mutual} if $j\in F_i$ if and only if $i\in F_j$ for all $i,j$.
For a mutual collection of orderings, we may consider the \textit{friendship graph} with vertex set $N$, 
where an (unweighted) edge connects mutual friends.
We will use standard terminology of graph theory when talking about hedonic games,
and, in particular, speak of cliques, trees, and cycles of agents.

\section{Hardness Results}\label{sec:claims}
\label{sec:results}
Let $\C$ be a polynomially representable class of hedonic games. For every stability concept $\alpha$
defined in Section~\ref{sec:prelim}, we will consider the 
following decision problem associated with $\C$.

\vspace{0.3em}
\noindent $\alpha$-\textsc{existence for \scC}\\
\textit{Instance:}\hspace{1pt} Game $\langle N, (\pref_i)_{i\in N}\rangle$ from $\C$ in its binary encoding. \\
\textit{Question}: Is there an $\alpha$-stable partition $\pi$ of $N$?
\vspace{0.35em}

\noindent To avoid difficulties with binary
representations that are very short,
we will assume that the binary encoding of $\langle N, (\pref_i)_{i\in N}\rangle$
lists the names of agents in $N$, and hence contains at least $|N|$ bits. 
Furthermore, when in the following theorems we assume that $\C$ contains 
various hedonic games $\langle N, (\pref_i)_{i\in N}\rangle$ 
derived from orderings $(\ge_i)_{i\in N}$, we require that such games (i.e., their binary descriptions) 
can be constructed in time polynomial in $|N|$; this property is necessary for our
hardness reductions to work in polynomial time and is satisfied by all classes of hedonic games
considered in this paper.

Our first result has mild assumptions and applies to a large number of classes $\C$.

\begin{theorem} 
	\label{thm:core-with-ties}
	\textsc{cr-existence for \scC} is \NP-hard if for all~$N$
     and every mutual collection of 
     orderings $(\ge_i)_{i\in N}$ in which each agent has at most 3 friends,
     there is a game $\langle N, (\pref_i)_{i\in N} \rangle\in\C$ 
     that is consistent on pairs, $\{\toxic01\}$-toxic and weakly $\{\toxic11, \toxic22\}$-toxic with respect to $(\ge_i)_{i\in N}$.
\end{theorem}

\begin{remark}
Under the same set of
conditions \textsc{sis-existence for~\scC} is also \NP-hard;  
we obtain a hardness result for \textsc{sns-existence for \scC} by
strengthening weak $\{\toxic11\}$-toxicity to $\{\toxic11\}$-toxicity.
\end{remark}

Effectively, Theorem~\ref{thm:core-with-ties} says that if agents are allowed to rank 
pairs as they wish, and if they do not have to like everyone, 
then finding a core-stable outcome is hard.

\setlength\intextsep{0pt}
\begin{wrapfigure}{r}{0pt}
	\hspace{-15pt}
	\includegraphics{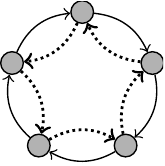}
\end{wrapfigure}
The assumptions are chosen so as to guarantee that a game like the pentagon displayed on the right has empty core. 
In this game, each agent has exactly two friends, the clockwise successor being preferred to the clockwise predecessor. 
All other agents are enemies. It can be checked that if agents' preferences 
satisfy weak $\{\toxic11, \toxic22\}$-toxicity then this game 
has empty core. We use the 9-player version of this game as a gadget
in our hardness reductions
(see Figure~\ref{fig:reduction-graphs}).

A similar result holds for solution concepts based on individual deviations.

\begin{theorem} 
	\label{thm:individual-stability}
	\textsc{ns-} and \textsc{is-existence for \scC} are \NP-complete if for all $N$
    and every \emph{strict} and mutual collection of orderings $(\ge_i)_{i\in N}$ in which each agent has at most 3 friends,
    there is a game $\langle N, (\pref_i)_{i\in N} \rangle\in\C$
    that is consistent on pairs and strictly $\{\toxic01,\toxic11, \toxic25\}$-toxic with respect to $(\ge_i)_{i\in N}$.
\end{theorem}

In the case of $\textsc{ns-existence}$, the theorem remains true even if the orderings $(\ge_i)_{i\in N}$ 
are strict and \textit{bipartite} (but not mutual), i.e.\ the friendship graph is bipartite. Thus, its 
conclusion also applies to $\textsc{ns-existence}$ for the stable marriage problem with unacceptabilities. For the 
case with ties allowed, this result is also obtained by Aziz~\shortcite{Aziz2013b}.

The reduction establishing Theorem~\ref{thm:core-with-ties} makes essential use of indifferences in the underlying 
orderings $(\ge_i)_{i\in N}$ (this is also the reason why it does not go through for the strict core). 
To cut off this cause of 
hardness, we need to make use of conditions on triangles. 

\begin{theorem} 
	\label{thm:strict-core-no-ties}
	\textsc{cr-} and \textsc{scr-existence for \scC} are 
	\NP-hard if for all $N$
         and every collection of \emph{strict} and mutual orderings $(\ge_i)_{i\in N}$
         in which each agent has at most 4 friends,
        there is a game $\langle N, (\pref_i)_{i\in N} \rangle\in\C$
        that is consistent on pairs, triangle-appreciating, monotone on triangles, $\{\toxic01 \}$-toxic, weakly $\{\toxic11,\toxic22,\toxic33\}$-toxic, and intolerant in triangles
        with respect to $(\ge_i)_{i\in N}$.
\end{theorem}

\begin{remark}
The same result holds for \textsc{sis-existence for \scC}. It applies to
\textsc{sns-existence for \scC} if we add $\{\toxic11\}$-toxicity and weak $\{\toxic21\}$-toxicity. It applies to
\textsc{ssns-existence for~\scC} if we add strict $\{\toxic01,\toxic11\}$-toxicity and weak $\{\toxic21\}$-toxicity.
\end{remark}

\begin{figure*}[t]
	\centering
	\includegraphics[width=1\textwidth]{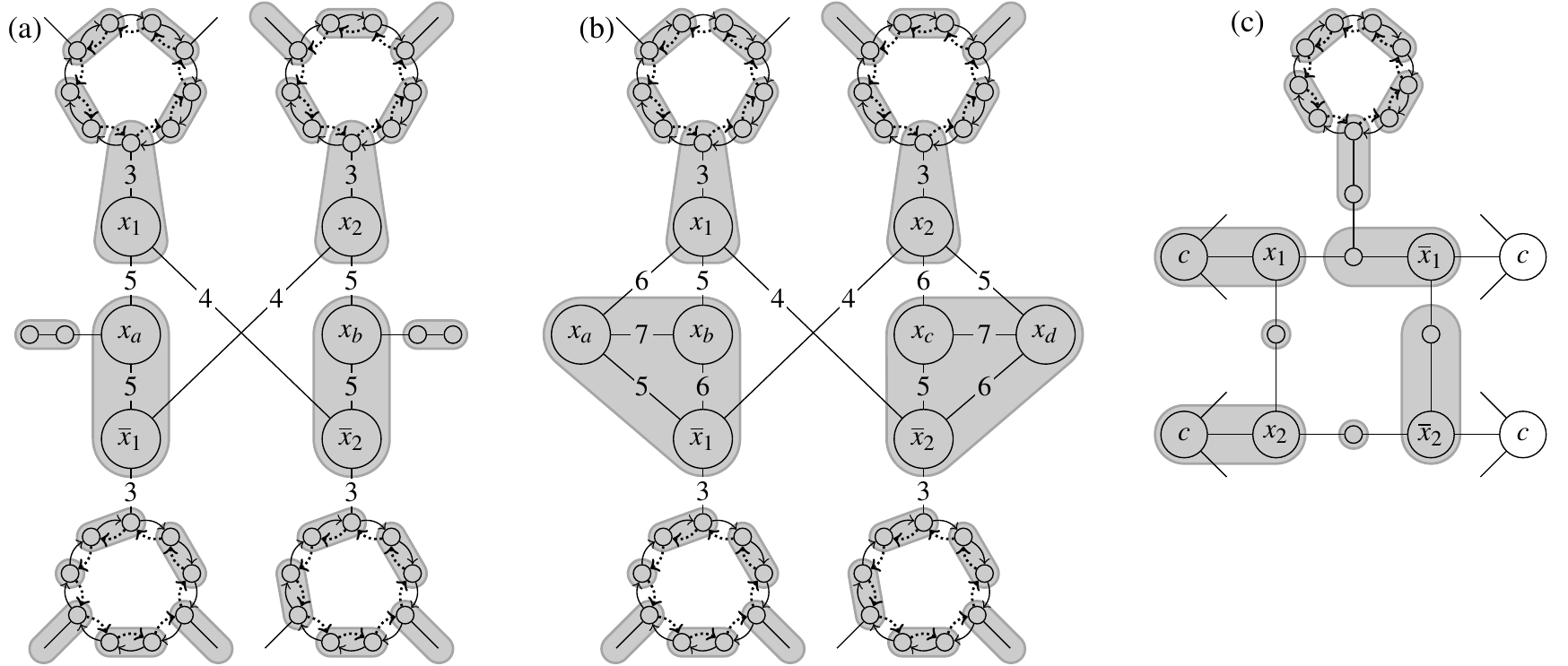}

	\caption{Graphical presentations of the \textsc{\protect\scalebox{0.85}{3}sat} reductions used. Figure (a) 
	is used in Theorem \ref{thm:core-with-ties}, (b) in Theorem \ref{thm:strict-core-no-ties}, and (c) in Theorem 
	\ref{thm:individual-stability}. 
	Agents not connected by an edge are enemies.
	The gray sets of agents indicate a stable partition $\pi$ in the hedonic game. The 
	$9$-gons form clause gadgets in (a) and (b) and a variable gadget in (c). On their own, these groups of 9 players do 
	not admit a stable outcome. So in (a) and (b), stability can only occur if for each clause one of its agents can be 
	connected to one of its (true) literals, i.e.\ if the underlying formula is satisfiable. In (c), each variable must 
	be set true or false for stability to occur.
		\vspace{-1.2em}}
	\label{fig:reduction-graphs}
\end{figure*}

\section{The Reductions}\label{sec:proofs}

The proofs of our results are by reduction from a restricted version of \threesat. The reduction behind 
Theorem~\ref{thm:core-with-ties} is inspired by an argument of Ronn~\shortcite{Ronn1990} 
showing that \textsc{stable-roommates} with ties is 
NP-complete. Theorem~\ref{thm:strict-core-no-ties} introduces triangles into this reduction to allow strict 
preferences.

We sketch the proof of Theorem \ref{thm:core-with-ties} but omit proofs of the other claims due to space 
constraints. The omitted arguments are similar to the one given, but more complicated due to \textsc{sns}-like 
stability concepts imposing little structure. 
Full proofs are given in the appendices A, B, and C. \\

\noindent \textsc{Proof of Theorem \ref{thm:core-with-ties} (sketch)}. 
We reduce from \textsc{(\scalebox{0.83}{3},b\scalebox{0.81}{2})-sat}, which is \threesat\ restricted to formulas 
in which each clause contains exactly $3$ literals, and each variable occurs exactly 
twice positively and twice negatively \cite{Berman2003}.

Given an instance formula $\varphi$ with variable set $X$ and clause set $C$, we construct the following agent set $N$:
\setlength\abovedisplayskip{3.5pt}
\setlength\belowdisplayskip{3.2pt}
\newcommand{\gci}{\operatorname{gc}\hspace{-0.6pt}\raisebox{0.4pt}{$_{\mathrm{i}}$}(\ell)}
\newcommand{\gce}{\operatorname{gc}\hspace{-0.6pt}\raisebox{0.4pt}{$_{\mathrm{e}}$}(\ell)}
\[\textstyle \bigcup_{x\in X} \{x_1, \overline x_1, x_2, \overline x_2, x_a, x_a', x_a'', x_b, x_b', x_b''\} \cup\hspace{0.3pt} \bigcup_{c\in C} \{ c_1, \scalebox{0.8}\dots, c_9 \}. \]
The four occurrences (two positive ones and two negative ones) of a variable $x\in X$ 
are called $x_1, x_2, \overline x_1, \overline x_2$. 
respectively. 
For a clause $c = \ell_1 \lor \ell_2 \lor \ell_3$, we write $c(\ell_1) := c_1$, $c(\ell_2) := c_4$, $c(\ell_3) = c_7$. Construct orderings $(\ge_i)_{i\in N}$ as follows:
\vspace{-2pt}
{\fontsize{9.5pt}{0.3em}\selectfont
\begingroup
\addtolength{\jot}{-0.2em}
\begin{alignat*}{10}
\overline x_1: &\enspace x_a > x_2 > c(\overline x_1) \enspace\enspace
& x_a: &\enspace x_1 \sim \overline x_1 > x_a' \enspace\enspace
& c_1: &\enspace \ell_1 > c_2 > c_9
\\
\overline x_2: &\enspace x_b > x_1 > c(\overline x_2) 
& x_b: &\enspace x_2 \sim \overline x_2 > x_b'
& c_4: &\enspace \ell_2 > c_5 > c_3
\\
x_1: &\enspace x_a > \overline x_2 > c(x_1)
& x_a': &\enspace x_a > x_a''
& c_7: &\enspace \ell_3 > c_8 > c_6
\\
x_2: &\enspace x_b > \overline x_1 > c(x_2)
& x_b': &\enspace x_b > x_b''
& c_i: &\enspace c_{i+1} > c_{i-1} \\
&
& x_a'': &\enspace x_a'
& x_b'': &\enspace x_b'
\end{alignat*}
\endgroup
}%

\noindent For each agent $i$ we have only listed $i$'s friends $F_i$, each friend being strictly better than $i$. 
Any agent not mentioned in $i$'s list is an enemy, i.e., an element of $E_i$. Figure~\ref{fig:reduction-graphs}(a) illustrates the orderings $(\ge_i)_{i\in N}$. 
Note that no agent has more than 3 friends, and that these orderings are mutual.

By the assumptions of Theorem \ref{thm:core-with-ties}, there is a poly-time many-one reduction that takes 
a formula~$\varphi$ as input and outputs the binary encoding of a game $G = \langle N, (\pref_i)_{i\in N} \rangle\in\C$ 
that is 
consistent on pairs, $\{\toxic01 \}$-toxic and weakly $\{\toxic11, \toxic22, \toxic33\}$-toxic with respect to the $(\ge_i)_{i\in N}$ given above. We show that $\varphi$ is satisfiable if and only if $G$ admits a \Core-stable partition.

Let $\A$ be a satisfying assignment of $\varphi$.
Take the partition
\begin{align*}
\pi &= \{ \{ \ell, c(\ell) \} : \ell \text{ a true variable occurrence} \} \\
& \enspace \cup \{ \{ x_a, \overline x_1 \}, \{ x_b, \overline x_2 \}  : x \in X \text{ set true in $\A$} \} \\
& \enspace \cup \{ \{ x_a, x_1 \}, \{ x_b, x_2 \}  : x \in X \text{ set false in $\A$} \} \\
& \enspace \cup \{ \{ x_a', x_a'' \}, \{ x_b', x_b'' \} : x \in X \} \\
& \enspace \cup \{ \{ c_i, c_{i+1} \}, \dots, \{ c_j \} : c \in C \}.
\end{align*}
In the last line we partition clause players that are not matched to true variables into pairs and singletons in some stable way as in Figure~1(a), see full proof for details.

We show that $\pi$ is \Core-stable in $G$. Since $\pi$ is \textsc{ir}, no singleton blocks. By consistency on pairs, it can be checked that no coalition of size 2 blocks. Now let $S$ be a coalition with $|S|\ge 3$. Consider the friendship graph on $N$ with friends connected by an edge. This graph has girth 6 and does not contain a cubic subgraph. If $S$ contained an isolated agent or a leaf (a member with at most 1 friend in $S$), then $S$ is not blocking by $\{\toxic01\}$-toxicity and weak $\{\toxic11\}$-toxicity. So $S$ contains a cycle and thus $|S| \ge 6$. Since $S$ is not cubic, there is an agent, matched in $\pi$, who has 2 friends in $S$ and so by weak $\{\toxic22\}$-toxicity is worse off in $S$, so $S$ does not block. Hence there are no blocking coalitions and $\pi$ is in the core.

Let $\pi$ be a \Core-stable partition of $G$. We sketch an argument giving a satisfying assignment of $\varphi$. Because the players 
$\{c_1,\dots,c_9\}$ of a clause are unstable on their own (toxicity limits coalitions within them to size $2$, no 
agent wants to be clockwise last in a coalition, and the number of members is odd), stability of $\pi$ 
implies that for each clause one of its players must be in a coalition with the literal connected to it. Define a propositional assignment 
$\A$ so that all literals in a coalition with a $c$-player are set true, and set other variables arbitrarily. This 
assignment is well-defined. Indeed, suppose variable $x$ is to be set both true and false. Then \textsc{wlog} 
either both $x_1$ and $\overline x_1$ or both $x_1$ and $\overline x_2$ are matched with their $c_1$. Either $\{ 
x_a, x_1 \}$ or $\{x_1, \overline x_2\}$ will then end up blocking, a contradiction. Clearly, $\A$ satisfies 
$\varphi$. \qed 

\section{Applications}\label{sec:apps}
\label{sec:applications}
Our NP-hardness results have implications
for many well-known classes of hedonic games.
In this section we 
briefly describe
some of these (see \cite{azizsavani} for details)
and check which of our conditions they satisfy.
In this way, we recover---and sometimes strengthen---a number of known hardness results
for these games. We also introduce 
five
new classes of games, and 
show how our framework allows us to deduce hardness results for them with ease. See the appendix for more details.

\vspace{2pt}

\class{Individually Rational Coalition Lists (IRCL)} 
Balles\-ter~\shortcite{Ballester2004} proposes to represent a hedonic game by 
listing the agent preferences $\pref_i$ explicitly from best to worst, 
but cutting the list off after the entry $\{i\}$. This 
representation is complete, but not always succinct. Ballester proves that 
for $\alpha\in\{$\Core,\NS,{\IS}$\}$
deciding $\alpha$-{\sc existence} 
is NP-complete under this representation. 
We deduce these results by considering IRCLs that list the pairs $\{ i,j \}$ for $j\in F_i$. Since 
Theorems~\ref{thm:core-with-ties} and~\ref{thm:individual-stability} apply even if each agent has only $3$ friends, we 
therefore have a hardness result for $\alpha\in\{$\SNS,\SIS,\Core,\NS,{\IS}$\}$ even if the list of each agent 
includes at most $3$ entries, each of which is a pair. A similar result is shown by 
Deineko and Woeginger~\shortcite{deineko2013two}. They 
prove that \Core-\textsc{existence for ircl} is hard even for lists of length $2$, with entries being coalitions of 
size $3$. Theorem~\ref{thm:strict-core-no-ties} applies if we allow lists up to length $9$, which can encode a 
triangle-appreciating game where agents have up to 4 friends.

\class{Hedonic Coalition Nets} 
Elkind and Wooldridge~\shortcite{Elkind2009} study a rule-based representation for hedonic games in which 
agents' preferences are described by weighted boolean formulas. It can be shown that polynomial size nets are 
sufficient to describe, for any collection of orderings $(\ge_i)_{i\in N}$, 
a game satisfying all our conditions, implying hardness of $\alpha$-\textsc{existence}
for all $\alpha$ considered in this work. This is perhaps not surprising:
while Elkind and Wooldridge only establish the hardness of \Core-\textsc{existence}
in their work, they show that one can compile an IRCL representation into a hedonic coalition net
representation with at most 
polynomial overhead. Because our hardness 
results hold even if each player is only allowed 3 or 4 
friends, we can say in addition that $\alpha$-\textsc{existence} for hedonic coalition nets remains hard even if we restrict 
each player's preferences to be described by at most 4 or 5 formulas, and even if the weights of these formulas are given in unary.

\class{Stable Roommates}
The reduction behind Theorem~\ref{thm:core-with-ties} is a modified version of 
Ronn's construction showing that \textsc{cr-existence for srt}, the stable roommate problem with ties, is NP-complete \cite{Ronn1990}.
It is thus no surprise that the class of stable 
roommate problems, considered as hedonic games in which sets with $3$ or more members are unacceptable, fulfills the 
conditions of Theorem~\ref{thm:core-with-ties} (but note that this formulation corresponds to \textsc{srti}, not \textsc{srt}). Indeed, \Core-\textsc{existence for srti} remains hard even if the preference list of each agent has length at most 3, and by Theorem~\ref{thm:individual-stability} this is also true of \textsc{ns-} and \textsc{is-existence}. Now, consider a generalization of \textsc{stable-roommates} where 
rooms have capacity $1$, $2$, or $3$, and rooms with capacity $3$ are generally preferred because they are cheaper per 
person. Then it can be checked that the conditions of 
Theorems~\ref{thm:individual-stability} and~\ref{thm:strict-core-no-ties} are satisfied, giving hardness of $\alpha$-{\sc existence}
for all $\alpha$ for this model.

\class{Stable Marriages}
A version of Theorem~\ref{thm:individual-stability} implies that \textsc{ns-existence for smi}, the stable marriage 
problem with incomplete lists, is NP-complete. This extends the NP-completeness result for \textsc{smti} obtained 
by \cite{Aziz2013b}. Aziz also notes that it is possible to embed \textsc{smti} into other 
classes~$\C$ of hedonic games, and thus to deduce hardness of \textsc{ns-existence for \scC};
this observation provides an (alternative) method of deriving hardness results 
for several classes of hedonic games. 

\class{$\mathcal{W}$-preferences} 
Cechl\'{a}rov\'{a} and Hajdukov\'{a}~\shortcite{Cechlarova2004} 
consider hedonic games where each agent first ranks all other agents
and then compares coalitions
based on their worst member under this ranking. 
Clearly, the game so obtained is consistent on pairs and strictly $\{\toxic{k}{1} \}$-toxic for all $k$. It follows that (with ties allowed) \Core-\textsc{existence} is NP-hard by Theorem~\ref{thm:core-with-ties} (a result first obtained by \cite{Cechlarova2004}), and that 
\NS- and \IS-\textsc{existence} are NP-complete by Theorem~\ref{thm:individual-stability}, 
first shown by Aziz {\em et al.}~\shortcite{Aziz2012a}. 
\NS- and \IS-\textsc{existence} are hard even if preferences are strict; 
the latter result was previously unknown.

\class{$\mathcal{W}\mathcal{B}$-preferences} 
Noting that agents in $\mathcal{W}$-hedonic games are extremely pessimistic, 
Cechl\'{a}rov\'{a} and Hajdukov\'{a}~\shortcite{Cechlarova2004a} 
propose a compromise: Agents still rank coalitions according to their worst member, but break ties in favor of the 
coalition with better best member. Again the game obtained is consistent on pairs and strictly 
$\{\toxic{k}{1} \}$-toxic for all $k$, so \Core-, \NS-, and \IS-{\sc existence} are hard. 

\class{W- and B-hedonic games}
In these two classes of games \cite{Aziz2012a,Aziz2013}, agents rank coalitions according to their worst or best member, but coalitions containing an enemy are not individually rational. As for $\mathcal W$-preferences, we see that \Core-, \NS-, and \IS-{\sc existence} are hard.
\\

\noindent In all of the following classes of games,
agents first assign cardinal utilities $v_i(j)\in\mathbb R$ to all agents in $N=\{1,\dots,n\}$, and then
lift these utilities to coalitions (e.g., by computing the sum or average
of the utilities of coalition members). 

The following method
of constructing integer-valued functions $v_i: N\to\mathbb R$ from orderings $(\ge_i)_{i\in N}$
will be used repeatedly: given $x,y\in {\mathbb Z}$, we set
$v_i(i)=0$, $v_i(j)=x$ for $j\in E_i$ and let $y\le v_i(j)\le y+n$
for $j\in F_i$ so that for each $j,k\in F_i$ we have $v_i(j)\ge v_i(k)$ iff $j\ge_i k$ 
(this is accomplished by assigning utility $y+k+1$ to friends
at the $k$-th `preference level'). 
We refer to such utilities
as $\llbracket \hspace{0.2pt}x,y\hspace{0.2pt}\rrbracket$-utilities. 

\class{Additively Separable Games (ASGs)} 
In these games, preferences are given by
$S\pref_i T$ iff $\sum_{j\in S} v_i(j) \ge \sum_{j\in T} v_i(j)$. 
This class of games 
satisfies all our theorems, so $\alpha$-\textsc{existence} 
is hard for all $\alpha$ we consider. 
Indeed, given $N = \{1,\dots,n\}$ and $(\ge_i)_{i\in N}$, we 
consider the ASG with $\llbracket-(n$\raisebox{-1pt}{$^2$}$+2n),4\hspace{0.2pt}\rrbracket$-utilities.
Then a coalition 
containing an enemy of $i$ is not individually rational for~$i$, so this game is 
strictly $\{\toxic{k}{1} \}$-toxic for all $k$, 
and it is obviously consistent on pairs, triangle-appreciating and monotone on triangles.
$\alpha$-\textsc{existence} remains hard even if players are allowed 
at most 3 or $4$ friends (depending on $\alpha$), so for ASGs, it remains hard even if $v_i(j)$ is positive for at most $3$ or $4$ 
agents~$j$. This improves on the reduction in \cite{Sung2010}, where agents have up to $11$ friends.

\class{Fractional Hedonic Games (FHGs)}
This class of games was recently proposed by Aziz {\em et al.}~\shortcite{Aziz2014a}. 
Preferences are given by $S\pref_i T$ iff 
$1/|S| \sum_{j\in S} v_i(j) \ge 1/|T| \sum_{j\in T} v_i(j)$. 
Brandl {\em et al.}~\shortcite{Brandl2015} 
have shown hardness of \mbox{$\Core$-,} $\NS$-, and $\IS$-\textsc{existence}. 
We recover these results and complement them by   
showing hardness of $\SSNS$-, $\SNS$-, $\SIS$- and $\SC\textsc{-existence}$; 
all these results hold even if the underlying preferences are strict. 
FHGs with $\llbracket-(n$\raisebox{-1pt}{$^2$}$+5n), 5\hspace{0.3pt}\rrbracket$-utilities
satisfy all of our properties; choosing $y=5$
ensures triangle-appreciation.

\class{Social FHGs}
An FHG is \textit{social} if agents' utilities for each other are non-negative.
Theorem~\ref{thm:strict-core-no-ties} applies to the class of social FHGs. 
Indeed, given $(\ge_i)_{i\in N}$
we can construct a social FHG with $\llbracket0,7n\rrbracket$-utilities.
Toxicity follows from $v_i(j)\ge 7n$ for $j\in F_i$, and other properties
can be checked as for FHGs. 
To ensure that our framework applies to social FHGs, we carefully crafted our constructions 
to only require weak toxicity whenever possible.

\smallskip

\noindent
The next 
five
classes of hedonic games are based on fairly intuitive ways
of deriving utilities for coalitions from utilities for individual players;
however, to the best of our knowledge we are the first to consider the computational
complexity of stability-related probems for these games (median games have been suggested
by \cite{Hajdukova2006} as an interesting topic; the other four classes appear
to be entirely new). 

\class{Median Games}
Agents evaluate coalitions according to their 
median value, which in odd-size coalitions is the middle element, and in even-size coalitions is the mean of the 
middle two elements. Median games with $\llbracket0,5\rrbracket$-utilities satisfy 
Theorem~\ref{thm:strict-core-no-ties}. Notice that
in this construction $v_i(j)$ are non-negative, so hardness holds even for `social median games' with non-negative underlying utilities. There are various other ways of defining median games. In particular, we can use a purely ordinal version by taking the worse of the middle two players in even-sized coalitions, satisfying Theorem~\ref{thm:core-with-ties}; if agents take either the ordinal or cardinal median of the coalition $S\setminus\{i\}$ then both Theorems~\ref{thm:core-with-ties} and~\ref{thm:individual-stability} apply.

\class{Geometric Mean Games}
In these games agents evaluate coalitions according 
to the geometric mean \scalebox{0.9}{$\sqrt[|S|]{\prod v_i(j)}$} of member 
utilities. We obtain the same hardness results as for FHGs by taking logs.

\class{Nash Product Games}
This is the class of games that are `multiplicatively separable'; agents evaluate coalitions according to $\prod_{j\in S} v_i(j)$. As 
far as hardness is concerned these games behave identically to additively separable games, again by taking logs.

\class{Midrange ($\frac12 \mathcal B + \frac12\mathcal W$)}
In this case, agents evaluate a coalition by averaging the maximum and minimum utility in it. 
With $\llbracket-3n,1 \rrbracket$-utilities, these games are strictly $\{\toxic{k}{1} \}$-toxic for all $k$ and consistent on pairs, so Theorems~\ref{thm:core-with-ties} and~\ref{thm:individual-stability} apply. 

\class{$r$-Approval}
Starting with cardinal utilities, sum the (up to) $r$ highest elements of a coalition. If $r\ge 3$, then games with $\llbracket -6r n,4\rrbracket$-utilities satisfy the conditions of Theorems ~\ref{thm:core-with-ties} and~\ref{thm:individual-stability}. If $r\ge 4$, they satisfy the conditions of Theorem~\ref{thm:strict-core-no-ties}.

\section{Conclusions}\label{sec:conclusions}
We have developed a framework that enables us to prove NP-hardness of $\alpha$-{\sc existence for \scC}
for many choices of $\alpha$ and~$\C$. Our results show that problems in this family tend to be hard
even for representation formalisms with very limited expressivity, and, moreover, are unlikely
to admit an efficient parametrized algorithm for many natural choices of parameter
(such as length and coalition size in the IRCL representation or number of formulas per agent in the hedonic coalition nets
representation). However, they also indicate which features of hedonic games may lead to tractability
of stability-related problems. In particular, restricting the number of different 
`preference intensities' (e.g., the range of $v_i(j)$ in ASGs, FHGs, and median games) 
rules out consistency on pairs, so one may hope for easiness results when this number is small.

While we focused on the problem of checking whether a stable partition exists,
another important stability-related problem is checking whether a specific 
partition is stable. This problem is in P for \IS\ and \NS\ for all
classes of hedonic games considered here, simply because the number of possible
deviations is polynomially bounded; however, for notions of stability that are 
based on group deviations it is often coNP-complete. It would be interesting to extend our framework
to handle this problem as well.

Since verifying stability is often hard, 
$\alpha$-{\sc existence for~\scC} is usually not known to be in NP for stability notions based on group deviations. Thus most of 
our hardness results do not have a tight
complexity upper bound.       
For all representation formalisms we consider, these problems are in $\Sigma_2^p$,
and $\Core$-{\sc existence for ASGs} is known to be complete for this complexity class \cite{Woeginger2013}.
A natural open question is whether our framework 
can be extended from NP-hardness proofs to $\Sigma_2^p$-hardness proofs.

\subsubsection*{Acknowledgements}
We thank the anonymous reviewers from IJCAI and CoopMAS for their helpful comments and pointers to the literature.

\bibliographystyle{named}
\bibliography{simple-causes-bibliography}

\newpage
\appendix

\section{Proof of Theorem 1}
\begin{theorem-non}
	\textsc{sis-} and \textsc{cr-existence} $\langle$\textsc{sns-existence}$\rangle$ \textsc{for \scC}  is NP-hard if for each $N$
	and every collection of 
	orderings $(\ge_i)_{i\in N}$ 
	there is a game $\langle N, (\pref_i)_{i\in N} \rangle\in\C$ 
	that is consistent on pairs, $\{\toxic01\}$-toxic $\langle\{\toxic01,\toxic11\}$-toxic$\rangle$ and weakly $\{\toxic11, \toxic22\}$-toxic with respect to $(\ge_i)_{i\in N}$.
\end{theorem-non}

We will prove both statements in the theorem together, with the proof for \textsc{sns-existence} added in $\langle\rangle$-brackets. Proof by reduction from (3,B2)-SAT (each clause contains exactly 3 literals, each variable occurs exactly twice positively and twice negatively).

Given a formula $\varphi$ we denote
\begin{enumerate}
	\item[(a)] for any variable $x$ its four occurrences by $x_1$, $x_2$, $\overline x_1$, $\overline x_2$,
	\item[(b)] for any variable occurence $\ell$, $c(\ell)$ is the clause that $\ell$ occurs in.
\end{enumerate}

Given an instance $\varphi$ of (3,B2)-SAT, take the following set $N$ of agents, with 9 agents per clause and 10 agents per variable.
\begin{align*}
	N = &\{ c_1,\dots,c_9 \mid c \in \text{Clauses}(\varphi) \} \\
	& \cup \{ x_1, x_2, \overline x_1, \overline x_2, x_a, x_a', x_a'', x_b, x_b', x_b'' \mid x \in\text{Var}(\varphi) \}.
\end{align*}

For a clause $c = \ell_1\lor\ell_2\lor\ell_3$, where $\ell_i$ is a variable occurrence, we will connect $c_1$ with $\ell_1$, $c_4$ with $\ell_2$, and $c_7$ with $\ell_3$. For $c_1$, $c_4$, $c_7$, we call the variable occurrence connected with them `its literal'. For a variable occurrence $\ell$, we will call the $c_i$ player connected with it `its clause player' and denote it by $c(\ell)$. We generate the following orderings, which only show the friends of each player. The last entry of player $i$'s list is strictly better than $i$ who is strictly better than all players not mentioned.
\begin{align*}
c_1 &: \ell_1 > c_2 > c_9 \\
c_4 &: \ell_2 > c_5 > c_3 \\
c_7 &: \ell_3 > c_8 > c_6 \\
c_i &: c_{i+1} > c_{i-1} \qquad\text{{\small for $i \neq 1,4,7$ with subscripts mod $9$}} \\
x_1 &: x_a > \overline x_2 > c(x_1) \\
\overline x_1 &: x_a > x_2 > c(\overline x_1) \\
x_2 &: x_b > \overline x_1 > c(x_2) \\
\overline x_2 &: x_b > x_1 > c(\overline x_2) \\	  
x_a &: x_1 \sim \overline x_1 > x_a' \\
x_a' &: x_a > x_a'' \\
x_a'' &: x_a' \\
x_b &: x_2 \sim \overline x_2 > x_b' \\
x_b' &: x_b > x_b'' \\
x_b'' &: x_b'
\end{align*}

\begin{figure}[ht]
	\scalebox{0.8}{\includegraphics{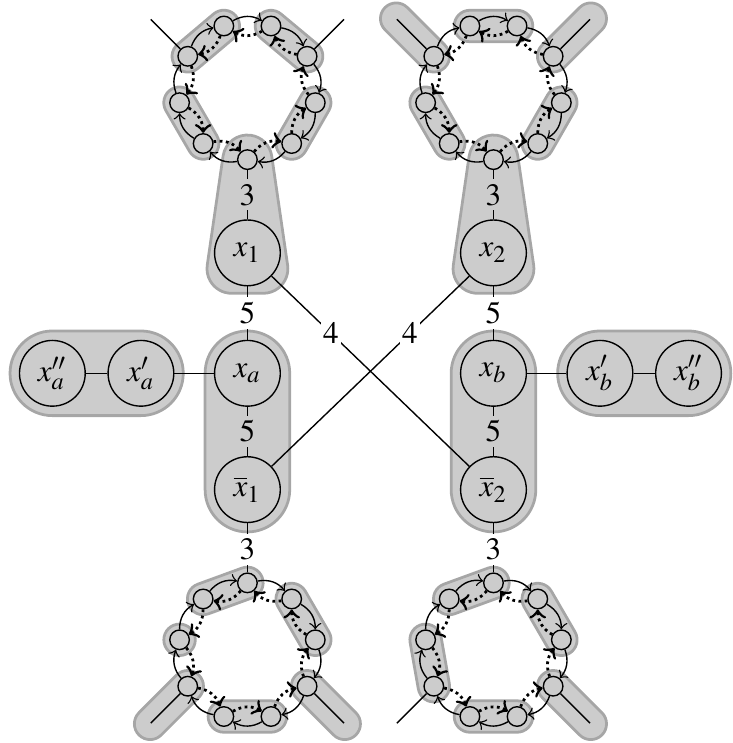}}
\end{figure}

\noindent Notice the following facts:
\begin{itemize}
	\item No player has more than 3 friends.
	\item The orderings are mutual: $j$ is $i$'s friend if and only if $i$ is $j$'s friend.
	\item The friendship graph has girth 6. In particular, it is triangle-free.
\end{itemize}
By the assumptions of the theorem, we can in polynomial time find a hedonic game which is consistent on pairs, \toxicc01-toxic $\langle\{ \toxic01,\toxic11 \}$-toxic$\rangle$, and weakly $\{\toxic11,\toxic22\}$-toxic with respect to the orderings given. We show that $\varphi$ is satisfiable if and only if this game admits a core- and SIS-stable $\langle$SNS-stable$\rangle$  partition.

\subsubsection{Satisfiable $\Rightarrow$ Stable.}

Suppose $\mathcal A$ is a satisfying assignment for $\varphi$. Construct the following partition $\pi$ of $N$, which consists only of pairs and singletons:
\begin{itemize}
	\item Match true variable occurrences with their clause player.
	\item If $x$ is false then $\{ x_1, x_a \}, \{ x_2, x_b \}\in \pi$.\\
	If $x$ is true then $\{ \overline x_1, x_a \}, \{\overline x_2, x_b \} \in \pi$.
	\item $\{ x_a', x_a'' \} \in \pi$ and $\{ x_b', x_b'' \}\in\pi$.
	\item For each clause, match the non-matched $c$-players in some stable way:
	\begin{itemize}
		\item If exactly 1 player is matched with a variable occurrence, say $c_1$ is matched, then we take $\{c_2, c_3\}, \{c_4,c_5\}, \{c_6,c_7\},\{c_8,c_9\}\in\pi$.
		\item If exactly 2 players are matched with a variable occurrence, say $c_1$ and $c_4$ are matched, then $\{c_2, c_3\}, \{c_5\}, \{c_6,c_7\},\{c_8,c_9\}\in\pi$.
		\item If all 3 players are matched with a variable occurrence, then $\{c_2, c_3\}, \{c_5,c_6\}, \{c_8,c_9\}\in\pi$.
	\end{itemize}
	Each of these three cases is illustrated in the drawing above.
\end{itemize}

Suppose that $\pi \overset H \longto \pi'$ for some $H\inn N$ with $\pi'\succ_i\pi$ for all $i\in H$. We will prove that this is not an SIS-deviation $\langle$not an SNS-deviation$\rangle$. Note that if there are no SIS-deviations, then $\pi$ is also core-stable.

Some terminology and observations:
\begin{itemize}
	\item an agent is \textit{matched} if it is in a pair in $\pi$, and \textit{lonely} if it is in a singleton in $\pi$ (these terms always refer to $\pi$ and never to $\pi'$).
	\item an agent is a \textit{deviator} if it is in $H$, and a \textit{non-deviator} otherwise.
	\item no 2 enemies are together in $\pi$.
	\item any 2 lonely players have distance at least 5, which in the following will mean that at most 1 player in a coalition considered below can be lonely. 
\end{itemize}

By definition of $\pi \overset H \longto \pi'$, no two non-deviators can be in the same coalition in $\pi'$ if they weren't together in $\pi$ already. This fact will be used often and not particularly stressed in the following argument.

\begin{lemma}
	\label{lemma:2-friends}
	No matched deviator $i$ has exactly $2$ friends in $\pi'(i)$.
\end{lemma}

\noindent
\textit{Proof.} Suppose $i$ is a matched deviator and ends up in a coalition $S\in\pi'$ which includes exactly 2 friends of $i$. Since $i$ prefers $\pi'$ to $\pi$, by weak \toxicc22-toxicity, $S$ includes at most 1 enemy of $i$, so $|S|$ is either 3 or 4.

	\[ \xymatrix{ j \ar@{-}[r] & i \ar@{-}[r] & k  } \]

If $|S|$ is 3, say $S = \{ i,j,k \}$, then since the game is triangle-free and friendship is mutual, each of $j$ and $k$ have 1 friend (namely $i$) and 1 enemy in $S$. Since $j$ and $k$ are enemies, they are not together in $\pi$ and hence they cannot be both non-deviators (and still end up in the same coalition in $\pi'$). Say $j$ is a deviator. Then $j$ cannot be matched by weak \toxicc11-toxicity ($j$ has an enemy in $S$). So $j$ is lonely and hence $k$ is not lonely since lonely players are far apart. Now $k$ is made worse off by the deviation (since $k$ is not lonely and by weak \toxic11-toxicity), so the deviation is not SIS. $\langle$With \toxicc11-toxicity, $j$ is not made happier by the deviation, so not an SNS deviation.$\rangle$

	\[ \xymatrix@R-1pc{ j \ar@{-}[r] & i \ar@{-}[r] & k \\ &&\ell \ar@{-}[u]  } \]

Suppose $|S|$ is 4, say $S = \{ i, j, k, \ell \}$ where $\ell$ is an enemy of $i$. Suppose first that $\ell$ had no friends in $S$. Then $\ell$ is not better off under $\pi'$ by \toxicc01-toxicity and thus is a non-deviator. Since $j$ and $k$ were not previously together with $\ell$, they must both be deviators. If either of them was matched, they'd now be unhappy by weak \toxicc11-toxicity. Since they are deviators, they cannot be unhappy, and hence both were lonely. But then $i$ is friends with 2 lonely players, contradiction. Hence $\ell$ has a friend in $S$. Since there are no 4-cycles in the game, $\ell$ cannot be friends with both $j$ and $k$, and thus must be friends with exactly 1 of $j$ and $k$, say $k$. So $j$ and $\ell$ are enemies. Since they are not together under $\pi$, at least 1 of them must be a deviator, say $j$. Since there are enemies in $S$, $j$ must have been lonely, and hence $\ell$ is not lonely. Now $\ell$ is made worse off by the deviation (since $\ell$ is not lonely), so the deviation is not SIS. $\langle$With \toxicc11-toxicity, $j$ is not made happier by the deviation, so not SNS.$\rangle$ \qed

\begin{lemma}
	No matched deviator $i$ has exactly $3$ friends in $\pi'(i)$.
\end{lemma}

\begin{proof}
	Suppose $i$ is a matched deviator and ends up in a coalition $S\in\pi'$ which includes exactly 3 friends of $i$, called $j$, $k$, $\ell$ who are (necessarily) pairwise enemies. Thus at least 2 of the 3 friends, say $j$ and $k$, must deviate, and thus at least 1 of these 3 friends is a matched deviator, say $j$. By weak $\{ \toxic11, \toxic22 \}$-toxicity, $j$ must have at least 3 friends in $S$, including $i$. Apart from $i$, 1 more friend of $j$ must deviate, called $m$. Now one of $m$ and $k$ must be matched, so is a matched deviator, and hence must have 3 friends in $S$. Hence $S$ contains at least 3 players with 3 friends in $S$. Also $|S| \ge 5$, and hence $S$ cannot include matched deviators who only have 2 or fewer friends in $S$, by toxicity. Call a matched player in $S$ with at most 2 friends in $S$ \textit{sad}. Then $S$ consists of lonely players, players with 3 friends in $S$, and at most 2 sad players (who are non-deviators). If there are 2 sad players, they must be together in $\pi$.
	
	Now if a player of type $x_a$ is in $S$, then either $x_a$ is sad or $x_a'\in S$ and $x_a'$ is sad. Similarly for $x_b$. Further, if a player of type $c_1$/$c_4$/$c_7$ is in $S$, then either he is sad, or a neighbour of it is in $S$ and sad. Therefore at most 1 player of type $x_a$/$x_b$/$c_1$/$c_4$/$c_7$ can be in $S$. Suppose a literal player $x_1$ is in $S$. Then either $x_1$ is sad, or else both $c(x_1)\in S$ and $x_a\in S$, but we previously said that there is at most 1 of these players in $S$, so this is impossible. Thus all literal players in $S$ must be sad, and so there is at most 1 literal player in $S$. Hence there are at most 2 degree-3 players in $S$, contradicting our observation before that there are at least 3.
\end{proof}

Now suppose $i$ is a matched deviator (with $\{i,j\}\in \pi$). Then by the lemmas, $i$ must end up in a coalition $S\in\pi'$ which includes exactly 1 friend of $i$. By weak $\toxicc11$-toxicity, $S$ must be a pair, $S = \{i,k\}$, say. Since $\pi' \succ_i \pi$ because $i$ is a deviator, we must have $k >_i j$ by consistency on pairs. Now by inspection it is seen that in $\pi$ there is no edge $\{ i,j \}$ that is strictly better than $\pi$ for both $i$ and $j$. Hence $k$ cannot be strictly better off in $\{i,k\}$ than $k$ is in $\pi$. So $k$ is not a deviator, and hence is either lonely or matched to a deviator.

We now go through each type $i$ of matched player and show that there is no edge $\{i,k\}$ satisfying these conditions: any preferred edge (if any) involves a partner who is matched to a non-deviator.

\begin{itemize}
	\item $x_a''$ and $x_b''$, $x_a$ and $x_b$, and false variable occurrences are in a favourite edge, so none of them are deviators.
	\item $x_a'$ and $x_b'$ have a strictly preferred edge with $x_a$ and $x_b$ respectively, but these are matched to false variable occurrences who we know are non-deviators.
	\item True variable occurrences: all preferred edges are partnered with a non-deviating player.
	\item Those $c_1$, $c_4$, or $c_7$ players matched to their literal player are in their favourite edge so not deviating.
	\item Those $c_i$ players together with $c_{i+1}$ are either in their favourite edge or (if they have a literal friend) that friend is false, so together with a non-deviator.
	\item Those $c_i$ players together with $c_{i-1}$ such that $c_{i+1}$ is in a pair where both members are confirmed non-deviators are themselves then clearly non-deviators. Such $c_i$ players exist: namely those preceding a $c$-player matched with their literal. If we repeat this observation, we find that all matched $c$-players are non-deviators.
\end{itemize}

Hence (together with the lemmas above) no matched player is a deviator.

Now consider a lonely player $i$. If $i$ deviates, then $i$ ends up in a coalition $S\in\pi'$ consisting of lonely players and possibly 2 players that are in a pair in $\pi$ (because no matched players are deviators). If $S$ consists entirely of lonely players, then by \toxicc01-toxicity, $i$ is not better off, so won't deviate. Hence $S$ also contains 2 players in a pair. At least 1 of these 2 is enemies with $i$ and is thus worse off by weak \toxicc11-toxicity, so this is not an SIS-deviation. $\langle$With \toxicc11-toxicity, $i$ is not better off in $S$, so not an SNS-deviation.$\rangle$ 

Thus we conclude that no player is a deviator. Hence $\pi$ is SIS-stable $\langle$SNS-stable$\rangle$.

\subsubsection{Stable $\Rightarrow$ Satisfiable.}

Suppose $\pi$ is a core-stable partition of the game. We show that $\varphi$ is then satisfiable.

\begin{lemma}
	The $9$ players of a clause cannot all be together in the same coalition in $\pi$.
\end{lemma}

\begin{proof}
	If they were, then $\{c_2, c_3\}$ would block by weak \toxicc22-toxicity.
\end{proof}

\begin{lemma}
	For any given clause, at least 1 of its players must be together with their literal player in $\pi$.
\end{lemma}

\begin{proof}
	Suppose not. Let $T = \{c_1,\dots,c_9\}$. Suppose 3 or more agents from $T$ are together in $S\in\pi$. Take an agent $c_i\in S$ such that $c_{i-1} \not\in S$ (this exists since $T\not\in\pi$). By weak $\{1,1\}$-toxicity, $c_i$ prefers any edge to $S$. Under $\pi$, $c_{i-1}$ can have at most one friend (because $c_i$ is committed to $S\not\ni c_{i-1}$ and we assumed that $c_{i-1}$ is not together with its literal player). Thus under $\pi$, $c_{i-1}$ is not better off than $\{ c_{i-2}, c_{i-1} \}$ (again by weak $\{1,1\}$-toxicity). It follows that $\{ c_{i-1}, c_i \}$ is blocking $\pi$, a contradiction. So at most two agents from $T$ are in the same coalition in $\pi$. Since $|T|$ is odd, there is a $c_i$ not together with any other agent from $T$ in $\pi$. By our assumption that $c_i$ has no literal friends, this $c_i$ has no friends in $\pi(c_i)$, so is not better off than being alone by $\{0,1\}$-toxicity. On the other hand, $c_{i-1}$ is not better off than $\{ c_{i-2}, c_{i-1} \}$. It follows that $\{ c_{i-1}, c_i \}$ is blocking $\pi$.
\end{proof}

\begin{lemma}
	No $c_i$ can have $3$ friends in $\pi$. A $c_i$ together with its literal has at most $1$ enemy in $\pi$.
\end{lemma}

\noindent
\textit{Proof.} Suppose $c_i$ does have 3 friends in $\pi$. 
\[
\xymatrix@R-1.7pc@C-1pc{
	&&& c_{i-2} \ar@{.}[dl] \\
	c_{i+1} \ar@{-}[dr] && c_{i-1} \ar@{-}[dl] \\
	& c_i \ar@{-}[dd] & \\
	\\
	& x &
}
\]
In particular, $c_i$ is together with $c_{i-1}$. If $c_i$ is also together with $c_{i-2}$, then each of $c_{i-1}$ and $c_{i-2}$ have at most 2 friends in $\pi$ (their degree is 2), but at least 2 enemies (the 2 other friends of $c_i$). Hence $\{ c_{i-1}, c_{i-2} \}$ blocks by weak $\{\toxic11, \toxic22\}$-toxicity. Hence $c_{i-2}$ is in a different coalition from $c_{i-1}$; thus $c_{i-2}$ has at most 1 friend in $\pi$. On the other hand $c_{i-1}$ has 1 friend but at least 2 enemies in $\pi$. Hence by weak \toxicc11-toxicity and consistency on pairs, $\{c_{i-1}, c_{i-2} \}$ blocks. So $c_i$ cannot have 3 friends in $\pi$.

Now to the second claim. We know $c_i$ has at most 2 friends in $\pi$. Suppose $c_i\in S\in \pi$, where $c_i$ has 2 enemies in $S$, and its literal player is part of $S$. By weak $\{\toxic11, \toxic22\}$-toxicity, $c_i$ prefers any edge to $S$. If $c_{i-1}\not\in S$, then $\{c_i, c_{i-1}\}$ blocks. If $c_{i-1}\in S$ but $c_{i-2}\not\in S$ then $\{c_i, c_{i-1}\}$ blocks. If both $c_{i-1}\in S$ and $c_{i-2}\in S$ then $|S|\ge 5$ ($S$ includes $c_i$, the literal, $c_{i-1}$, and two enemies of $c_i$), so that $\{c_{i-1},c_{i-2}\}$ blocks since they have at most 2 friends and at least 2 enemies. \qed

\begin{lemma}
	$x_1$ and $\overline x_1$ cannot both be together with their clause player.
\end{lemma}

\begin{proof}
	Suppose they are. Now if $x_a$ was together with both $x_1$ and $\overline x_1$ in $\pi$, this would mean that the associated clause players have more than 1 enemy which is impossible. So $x_a$ has at most 2 friends. But if $x_a$ had friends $x_a'$ and $x_1$ then the clause player of $x_1$ would have enemies $x_a$ and $x_a'$, which is impossible. So $x_a$ has either no or 1 friend in $\pi$.
	
	Suppose $x_a$ has as friend either $x_1$ or $\overline x_1$ in $\pi$. Then by weak \toxicc11-toxicity (since the clause player is an enemy), $\{x_a, x_a'\}$ blocks.
	
	Hence either $x_a$ has no friends in $\pi$, or its only friend is $x_a'$. Now if either $x_1$ or $\overline x_1$ has only 1 friend in $\pi$, then $\{x_a, x_1\}$ or $\{x_a, \overline x_1\}$ blocks. So both $x_1$ and $\overline x_1$ have 2 friends in $\pi$, so $x_1$ is together with $\overline x_2$ and $\overline x_1$ is together with $x_2$. Since the clause players of $x_1$ and $\overline x_1$ can have at most 1 enemy in $\pi$, it follows that $\overline x_2$ and $x_2$ have only 1 friend in $\pi$ (namely $x_1$ and $\overline x_1$ respectively). It follows that $x_b$ has at most 1 friend, namely possibly $x_b'$. Hence by weak \toxicc11-toxicity, $\{ x_b, x_2 \}$ blocks.
\end{proof}

\begin{lemma}
	$x_1$ and $\overline x_2$ cannot both be together with their clause player.
\end{lemma}

\begin{proof}
	Suppose they were. Since the clause players can have at most 1 enemy, $x_1$ and $\overline x_2$ cannot be in the same coalition in $\pi$. Now if $x_1$ and $\overline x_2$ both have only 1 friend (their clause player) in $\pi$, then $\{ x_1, \overline x_2 \}$ blocks by consistency on pairs. Otherwise, at least 1 of them, say $x_1$, has 2 friends in $\pi$, hence is together with $x_a$. Since $x_1$'s clause player can have at most 1 enemy, $x_a'$ is not in the same coalition. But then $\{x_a, x_a'\}$ blocks by weak \toxicc11-toxicity. 
\end{proof}

Define a propositional assignment $\mathcal A$ that sets literals that are in a coalition with their clause player true. By the last two lemmas, this is well-defined. By the lemma before, each clause has at least 1 literal that is set true by $\mathcal A$. Hence $\mathcal A$ satisfies $\varphi$.

\section{Proof of Theorem 2}
\subsection{Bipartite Case}

We call a collection of orderings $(\ge_i)_{i\in N}$ \textit{bipartite} if the friendship graph is bipartite, i.e.\ there is a partition $(N_1, N_2)$ of the agent set $N$ such that for each $i\in N_1$ we have $F_i\inn N_2$ and for each $i\in N_2$ we have $F_i \inn N_1$. The following theorem gives a hardness result even for bipartite preferences, and so it applies for example to the stable marriage case. This result works for Nash stability but not for individual stability (and it cannot, for individual stability is poly-time solvable for stable marriages). The next section will consider the non-bipartite case which also applies to individual stability.

\begin{theorem-non} 
	\textsc{ns-existence for \scC} is NP-complete if for all $N$
	and every \emph{strict bipartite} collection of orderings $(\ge_i)_{i\in N}$
	there is a game $\langle N, (\pref_i)_{i\in N} \rangle\in\C$
	that is consistent on pairs and strictly $\{\toxic01,\toxic11, \toxic25\}$-toxic with respect to $(\ge_i)_{i\in N}$.
\end{theorem-non}

The problem is in NP since a Nash-stable partition is a certificate: we can in polynomial time check for each player $i$ whether he wishes to deviate (this follows from the definition of a polynomially representable class).

We reduce from (3,B2)-SAT (each clause contains exactly 3 literals, each variable occurs exactly twice positively and twice negatively).

Given a formula $\varphi$ we denote
\begin{enumerate}
	\item[(a)] for any variable $x$ its four occurrences by $x_1$, $x_2$, $\overline x_1$, $\overline x_2$,
	\item[(b)] for any variable occurence $\ell$, $c(\ell)$ is the clause that $\ell$ occurs in.
\end{enumerate}

We will introduce 9 players for each variable
and 1 player for each clause.
\begin{align*}
N = &\{ x_{\text{stalker}}, x_{\text{main}}, x_{\text{pos}}, x_{\text{neg}}, x_{\text{garbage}} \mid x \in\on{Var}(\varphi) \}  \\
&\cup \{ x_1, x_2, \overline x_1, \overline x_2 \mid x \in\on{Var}(\varphi) \}  \\
&\cup \{ c \mid c \in\on{Clauses}(\varphi) \}.
\end{align*}

We take the following strict orderings, which only show the friends of each player. The last entry of player $i$'s list is strictly better than $i$ who is strictly better than all players not mentioned.
\begin{align*}
\color{red} x_{\text{stalker}} &:\quad x_{\text{main}} \\
\color{red} x_1 &:\quad c(x_1) > x_{\text{pos}} > x_{\text{main}} \\
\color{red} x_2 &:\quad c(x_2) > x_{\text{pos}} > x_{\text{garbage}} \\
\color{red} \overline x_1 &:\quad c(\overline x_1) > x_{\text{neg}} > x_{\text{main}} \\
\color{red} \overline x_2 &:\quad c(\overline x_2) > x_{\text{neg}} > x_{\text{garbage}} \\
\color{blue}c &:\quad \text{its three variable occurrences in any order} \\[-3pt]
\color{blue} x_{\text{main}} &:\quad x_1 > \overline x_1 \\
\color{blue}x_{\text{pos}} &:\quad x_1 > x_2 \\
\color{blue}x_{\text{neg}} &:\quad \overline x_1 > \overline x_2 \\
\color{blue}x_{\text{garbage}} &:\quad x_2 > \overline x_2.
\end{align*}
Notice that no one has more than 3 friends, and that these orderings are bipartite: all friends of a red player are blue, all friends of a blue player are red.

Let $G$ be any hedonic game that is consistent on pairs and strictly $\{\toxic01, \toxic11, \toxic25 \}$-toxic with respect to these orderings. This game has a Nash stable outcome if and only if $\varphi$ is satisfiable.
\vspace{-20pt}
\[  
\xymatrix{
	&&&\\
	c \ar@{}[dr]^(.07){}="a"^(.4){}="b" \ar@{-}_7 "a";"b"
	\ar@{}[ur]^(.07){}="a"^(.4){}="b" \ar@{-}^5 "a";"b"
	\ar@{-}[r]^6 & x_2 \ar@{-}[r]^2 \ar@{-}[d]_3 & x_{\text{garbage}} \ar@{-}[r]^1 & \overline x_2 \ar@{-}[d]^3 & c'  \ar@{}[dl]^(.07){}="a"^(.4){}="b" \ar@{-} "a";"b"
	\ar@{}[ul]^(.07){}="a"^(.4){}="b" \ar@{-} "a";"b"
	\ar@{-}[l] \\
	& x_{\text{pos}} \ar@{-}[d]_4  & & x_{\text{neg}} \ar@{-}[d]^4 \\
	c'' \ar@{}[dr]^(.07){}="a"^(.4){}="b" \ar@{-} "a";"b"
	\ar@{}[ur]^(.07){}="a"^(.4){}="b" \ar@{-} "a";"b"
	\ar@{-}[r] & x_1 \ar@{-}[r]_2 & x_{\text{main}} \ar@{-}[r]_1 & \overline x_1 & c''' \ar@{}[dl]^(.07){}="a"^(.4){}="b" \ar@{-} "a";"b"
	\ar@{}[ul]^(.07){}="a"^(.4){}="b" \ar@{-} "a";"b"
	\ar@{-}[l] \\
	&& x_{\text{stalker}} \ar[u] &&
}
\]

\subsubsection{Satisfiable $\Rightarrow$ NS-stable}
Suppose $\A$ is a satisfying assignment for $\varphi$. For a clause $c$, let $\on{true}(c)$ be the first variable occurrence in $c$ that is set true by $\A$ (this must exist since $\A$ satisfies $\varphi$).

We call a variable occurrence $\ell$ \textit{matched} if it is $\on{true}(c)$ for its clause $c$, and \textit{unmatched} otherwise.

We now describe a partition $\pi$ of $N$ that is Nash stable in $G$.
\begin{itemize}
	\item $\{x_{\text{stalker}}\}$ is always alone.
	\item $\{ c, \on{true}(c) \}$ forms a \textit{love marriage}.
	\item Suppose now that $x$ is true in $\A$.
	\begin{itemize}
		\item $\{ x_{\text{main}}, \overline x_1 \}$.
		\item $\{ x_{\text{neg}}, \overline x_2 \}$.
		\item If $x_1$ is matched and $x_2$ is matched, then $\{ x_{\text{pos}} \}$ and $\{x_{\text{garbage}}\}$.
		\item If $x_1$ is matched and $x_2$ is unmatched, then $\{ x_2, x_{\text{pos}} \}$ and $\{x_{\text{garbage}}\}$.
		\item If $x_1$ is unmatched and $x_2$ is matched, then $\{ x_1, x_{\text{pos}} \}$ and $\{x_{\text{garbage}}\}$.
		\item If $x_1$ is unmatched and $x_2$ is unmatched, then $\{ x_1, x_{\text{pos}} \}$ and $\{x_2, x_{\text{garbage}}\}$.
	\end{itemize}
	\item Suppose now that $x$ is false in $\A$. (everything is symmetric to the positive case)
	\begin{itemize}
		\item $\{ x_{\text{main}}, x_1 \}$.
		\item $\{ x_{\text{pos}}, x_2 \}$.
		\item If $\overline x_1$ is matched and $\overline x_2$ is matched, then $\{ x_{\text{neg}} \}$ and $\{x_{\text{garbage}}\}$.
		\item If $\overline x_1$ is matched and $\overline x_2$ is unmatched, then $\{ \overline x_2, x_{\text{neg}} \}$ and $\{x_{\text{garbage}}\}$.
		\item If $\overline x_1$ is unmatched and $\overline x_2$ is matched, then $\{ \overline x_1, x_{\text{neg}} \}$ and $\{x_{\text{garbage}}\}$.
		\item If $\overline x_1$ is unmatched and $\overline x_2$ is unmatched, then $\{ \overline x_1, x_{\text{neg}} \}$ and $\{\overline x_2, x_{\text{garbage}}\}$.
	\end{itemize}
\end{itemize}
The $\pi$ as above is Nash stable (note all coalitions have size at most 2, so it is a marriage matching). It is easily seen to be individually rational since no one is together with an enemy and consistency on pairs holds. Notice that because the underlying preferences are bipartite they are \textit{triangle-free} which means that no player ever wants to join a pair of players because at least one of them is an enemy and strict $\{\toxic01, \toxic11 \}$-toxicity holds. So any possible Nash deviation would involve a player joining a singleton coalition. The only players that are possibly single in $\pi$ are $x_{\text{pos}}, x_{\text{neg}}, x_{\text{garbage}}, x_{\text{stalker}}$. 
\begin{itemize}
	\item No one is friends with $x_{\text{stalker}}$ so no one can benefit by joining him.
	\item Player $x_{\text{pos}}$ is single only if both $x_1$ and $x_2$ are matched, and because they prefer their clause to $x_{\text{pos}}$, they do not benefit by joining $x_{\text{pos}}$. No one else is friends with $x_{\text{pos}}$, so they do not join either. 
	\item Similarly for $x_{\text{neg}}$.
	\item Player $x_{\text{garbage}}$ is single only in situations where his friends $x_2$ and $\overline x_2$ are together with either their clauses or the $x_{\text{pos}}, x_{\text{neg}}$ players. Both $x_2$ and $\overline x_2$ prefer this situation to joining $x_{\text{garbage}}$.
\end{itemize}
Hence no deviations are possible and thus $\pi$ is stable.

\subsubsection{NS-stable $\Rightarrow$ satisfiable} 

Suppose $\pi$ is NS-stable in $G$. 

\begin{lemma}
	All coalitions in $\pi$ have size 1 or 2.
\end{lemma}

\begin{proof}
	Remember that no player has more than 3 friends. Let $S\in \pi$ with $|S| \ge 3$. Then by toxicity and individual rationality of $\pi$, each player has either 3 friends in $S$, or 2 friends in $S$ but at most 4 enemies in $S$. So each member of $S$ has `degree' at least 2, and thus $S$ contains a cycle\footnote{We can use graph theory terminology by referring to a graph on $N$ with an edge between mutual friends. Note that since $x_{\text{stalker}}$ has only 1 friend, we must have $x_{\text{stalker}}\not\in S$, so we may pretend that all friendships are mutual so that all edges are indeed undirected.}. All cycles in $G$ have length 8 or more\footnote{Bipartiteness excludes cycles of odd length, so we need only check that there are no cycles of length 4 or 6. There are none. A shortest cycle is $x_{\text{main}} \to x_1 \to x_{\text{pos}}\to x_2 \to x_{\text{garbage}} \to \overline x_2 \to x_{\text{neg}} \to \overline x_1 \to x_{\text{main}}$ or cycles like $c \to x_1 \to x_{\text{pos}} \to x_2 \to c' \to y_2 \to y_{\text{pos}} \to y_1 \to c$.}, thus $|S| \ge 8$. If some member $i$ of $S$ had exactly 2 friends in $S$, then $i$ would have 5 enemies in contradiction to individual rationality by strict $\{\toxic25\}$-toxicity. Hence every $i\in S$ has exactly 3 friends in $S$. The only players who have 3 friends are variable occurrences $x_1,x_2,\overline x_1, \overline x_2$ and clauses $c$. If some variable occurrence were in $S$ then so would be its friends of types $x_{\text{main}}, x_{\text{pos}}, x_{\text{neg}}, x_{\text{garbage}}$ who themselves cannot have 3 friends. Hence $S$ can only contain clause players; but no two clause players are friends.
\end{proof}

Thus every player is either alone or together with exactly 1 friend. Consider $x_{\text{main}}$. If he is alone, then $x_{\text{stalker}}$ will want to join him. Since $\pi$ is Nash stable, this cannot happen. So $x_{\text{main}}$ is together with a friend, which is either $x_1$ or $\overline x_1$.

Define the following propositional assignment $\A$: 
\begin{align*}
\A(x) = \text{true} &\iff \{ x_{\text{main}}, \overline x_1 \} \in \pi \\
\A(x) = \text{false} &\iff \{ x_{\text{main}}, x_1 \} \in \pi
\end{align*}
By what we said above, this is well-defined. We will show that $\A$ satisfies $\varphi$.

Let $c$ be a clause. If $c$ is alone in $\pi$ then one of its literals joins $c$ (and is welcome to do so). Hence $c$ is not alone and thus together with one of its literals, say $\ell$. We show that $\ell$ is true under $\A$.

Suppose not and it is false. Of the two false literal occurrences of a variable, the first is in a pair with $x_{\text{main}}$ by definition of $\A$. Thus the second occurrence ($\ell$) is the one together with the clause; for concreteness suppose the situation is $\{ x_{\text{main}}, x_1 \}\in\pi$ and $\{ x_2, c \}\in \pi$. It then follows that $\{x_{\text{pos}}\}\in\pi$ because both friends of $x_{\text{pos}}$ are otherwise engaged. But then $x_1$ wants to join $x_{\text{pos}}$ (and is welcome to do so). This is a contradiction to $\pi$ being stable. Hence $\ell$ is true.

Therefore each clause contains a true literal under $\A$ and hence $\A$ satisfies $\varphi$. Thus $\varphi$ is satisfiable.

\subsection{Non-bipartite case}

\begin{theorem-non} 
	\textsc{ns-} and \textsc{is-existence for \scC} are NP-complete if for all $N$
	and every \emph{strict} collection of orderings $(\ge_i)_{i\in N}$
	there is a game $\langle N, (\pref_i)_{i\in N} \rangle\in\C$
	that is consistent on pairs and strictly $\{\toxic01,\toxic11, \toxic25\}$-toxic with respect to $(\ge_i)_{i\in N}$.
\end{theorem-non}

If we do not insist on bipartiteness, we can get a result for mutual preferences that also holds for IS. This is done by adding by a 9-gon. The argument that an NS-stable partition exists will be very similar to before, using triangle-freeness. Since the girth of the game continues to be at least 8, and since we have been careful that most deviations considered are actually IS-deviations, the remainder of the proof needs few adjustments.

In more detail:
\begin{align*}
N = & \{ x^1,\dots, x^9, x_{\text{stalker}}, x_{\text{main}}, x_{\text{pos}}, x_{\text{neg}}, x_{\text{garbage}} \mid x \in\on{Var}(\varphi) \} \\
&\cup \{ x_1, x_2, \overline x_1, \overline x_2 \mid x \in\on{Var}(\varphi) \} \\
&\cup \{ c \mid c \in\on{Clauses}(\varphi) \}.
\end{align*}

We take the following strict orderings:
\begin{align*}
x_{\text{main}} &:  x_{\text{stalker}} > x_1 > \overline x_1 \\
x_{\text{stalker}} &: x_{\text{main}} > x^1 \\
x^1 &: x_{\text{stalker}} > x^2 > x^9 \\
x^i &: x^{i+1} > x^{i-1} \quad\text{{\small with superscripts mod 9 ($i=2,\dots,9$)}}
\end{align*}
and everyone else as before. Given an assignment, the partition $\pi$ generated is the same as before, but instead of $\{x_{\text{stalker}}\} \in \pi$, we now take $\{x_{\text{stalker}}, x^1\}, \{x^2, x^3\}, \{x^4, x^5\}, \{x^6, x^7\}, \{x^8, x^9\}\in \pi$. Checking that this $\pi$ is NS proceeds as before.

Suppose the game has an IS partition $\pi$. Similar to before, all coalitions in $\pi$ have size 1 or 2. We can then see that we must have $\{x_{\text{stalker}}, x^1\} \in \pi$. Otherwise, there must be some single $\{x^i\}\in \pi$ and then $x^{i-1}$ will want to join $x^i$ and is welcome to do so. But now $x_{\text{main}}$ cannot be alone, else $x_{\text{stalker}}$ will want to join him and is welcome to do so. Hence $x_{\text{main}}$ is together with a friend. The rest of the argument can now proceed as before.

\section{Proof of Theorem 3}
\begin{theorem-non} 
	$\langle\textsc{sns-}\rangle$, $\llangle\textsc{ssns-}\rrangle$, \textsc{sis-}, \textsc{cr-} and \textsc{scr-existence for \scC} are 
	NP-hard if for all $N$
	and every collection of \emph{strict} orderings $(\ge_i)_{i\in N}$
	there is a game $\langle N, (\pref_i)_{i\in N} \rangle\in\C$
	that is consistent on pairs, $\{\toxic01 \}$-toxic, weakly $\{\toxic11,\toxic22,\toxic33\}$-toxic 
	
	$\langle \toxicc11$-toxic and weakly $\{\toxic21\}$-toxic$\rangle$, 
	
	$\llangle$strictly $\{\toxic01, \toxic11\}$-toxic and weakly $\{\toxic21\}$-toxic$\rrangle$, 
	
	intolerant in triangles, triangle-appreciating, and monotone on triangles with respect to $(\ge_i)_{i\in N}$.
\end{theorem-non}

We will prove all three statements in the theorem together. Proof by reduction from (3,B2)-SAT.

Given an instance $\varphi$ of (3,B2)-SAT, take the following set $N$ of agents, with 9 agents per clause and 8 agents per variable.
\begin{align*}
N = &\{ c_1,\dots,c_9 \mid c \in \text{Clauses}(\varphi) \} \\
&\cup \{ x_1, x_2, \overline x_1, \overline x_2, x_a, x_b,x_c,x_d \mid x \in\text{Var}(\varphi) \}.
\end{align*}
For a clause $c = \ell_1\lor\ell_2\lor\ell_3$, where $\ell_i$ is a variable occurrence, we will connect $c_1$ with $\ell_1$, $c_4$ with $\ell_2$, and $c_7$ with $\ell_3$. For $c_1$, $c_4$, $c_7$, we call the variable occurrence connected with them `its literal'. For a variable occurrence $\ell$, we will call the $c_i$ player connected with it `its clause player' and denote it by $c(\ell)$. We generate the following orderings, which only show the friends of each player. The last entry of player $i$'s list is strictly better than $i$ who is strictly better than all players not mentioned.

\begin{align*}
c_1 &: \ell_1 > c_2 > c_9 \\
c_4 &: \ell_2 > c_5 > c_3 \\
c_7 &: \ell_3 > c_8 > c_6 \\
c_i &: c_{i+1} > c_{i-1} \qquad\text{for $i \neq 1,4,7$ {\small with subscripts mod 9}} \\
x_1 &: x_a > x_b > \overline x_2 > c(x_1) \\
\overline x_1 &: x_b > x_a > x_2 > c(\overline x_1) \\
x_2 &: x_c > x_d > \overline x_1 > c(x_2) \\
\overline x_2 &: x_d > x_c > x_1 > c(\overline x_2) \\	  
x_a &: x_b > x_1 > \overline x_1 \\
x_b &: x_a > \overline x_1 > x_1 \\
x_c &: x_d > x_2 > \overline x_2 \\
x_d &: x_c > \overline x_2 > x_2 \\
\end{align*}

\begin{figure}[ht]
	\includegraphics{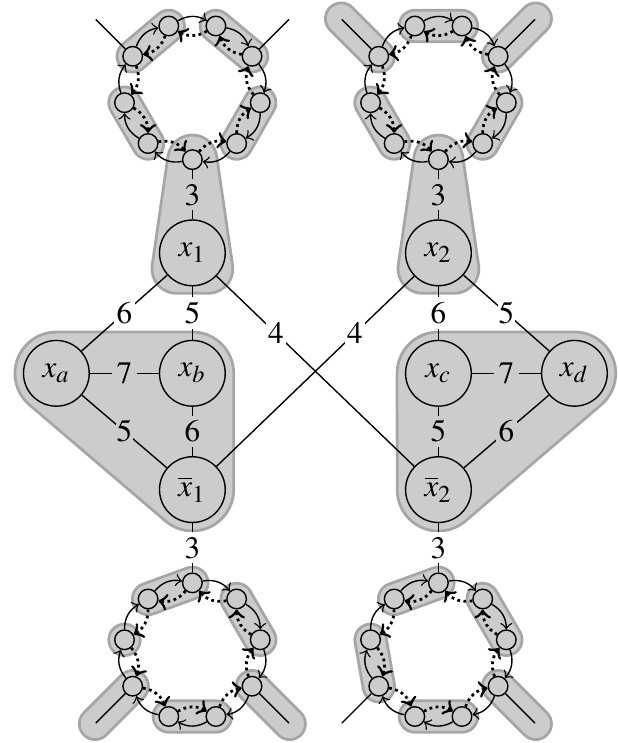}
\end{figure}

These orderings only show the friends of each player. The last entry of player $i$'s list is strictly better than $i$ who is strictly better than all players not mentioned.
Notice the following facts:
\begin{itemize}
	\item No player has more than 4 friends.
	\item The orderings are mutual: $j$ is $i$'s friend if and only if $i$ is $j$'s friend.
	\item The friendship graph has no chordless 4-cycles.
\end{itemize}
By the assumptions of the theorem, we can in polynomial time find a hedonic game which satisfies the relevant conditions as in the theorem statement. We show that $\varphi$ is satisfiable if and only if this game admits a core- and SIS-stable $\langle$SNS-stable$\rangle$ $\llangle$SSNS-stable$\rrangle$ partition.

\subsubsection{Satisfiable $\Rightarrow$ Stable.}

Suppose $\mathcal A$ is a satisfying assignment for $\varphi$. Construct the following partition $\pi$ of $N$, which consists of singletons, pairs, and triangles:
\begin{itemize}
	\item Match true variable occurrences with their clause player.
	\item If $x_1$ is false then $\{ x_1, x_a, x_b \}\in \pi$. If $\overline x_1$ is false then $\{ \overline x_1, x_a, x_b \}\in \pi$.
	\item If $x_2$ is false then $\{ x_2, x_c, x_d \}\in \pi$. If $\overline x_2$ is false then $\{ \overline x_2, x_c, x_d \}\in \pi$.
	\item For each clause, match the non-matched $c$-players in the same way as before.
\end{itemize}

Suppose that $\pi \overset H \longto \pi'$ for some $H\inn N$ with $\pi'\succ_i\pi$ $\llangle \pi'\pref_i\pi \rrangle$ for all $i\in H$. We will prove that this is not an SIS-deviation $\langle$not an SNS-deviation$\rangle$ $\llangle$not an SSNS-deviation$\rrangle$. Note that if there are no SIS-deviations, then $\pi$ is also core-stable.

Some terminology and observations:
\begin{itemize}
	\item an agent is \textit{lonely} if it is in a singleton in $\pi$ and \textit{matched} otherwise (these terms always refer to $\pi$ and never to $\pi'$). We also use `matched' for agents in triangles.
	\item an agent is a \textit{deviator} if it is in $H$, and a \textit{non-deviator} otherwise.
	\item no 2 enemies are matched in $\pi$.
	\item any 2 lonely players have distance 5, which in the following will mean that at most 1 player in a coalition considered below can be lonely. 
\end{itemize}

\begin{lemma}
	No matched deviator $i$ has exactly $2$ friends $j$ and $k$ in $\pi'$ where $j$ and $k$ are enemies.
\end{lemma}

\begin{proof}
	Essentially identical to Lemma \ref{lemma:2-friends}. Replace ``because there are no 4-cycles'' by ``because there are no chordless 4-cycles'', and $\llangle$modify remarks about SNS to also include the SSNS case$\rrangle$.
\end{proof}

\begin{lemma}
	No matched deviator $i$ has exactly $2$ friends $j$ and $k$ in $\pi'$ where $j$ and $k$ are friends.
\end{lemma}

\textit{Proof.} Suppose not. By weak \toxicc22-toxicity, $i$'s coalition $S$ in $\pi'$ has size either 3 or 4. $\langle$By weak \toxicc21-toxicity, $i$'s coalition $S$ in $\pi'$ has size 3.$\rangle$

	\[ \xymatrix@R-1pc@C-1pc{ & x_1 \ar@{-}[dr] \ar@{-}[dl] & \\ x_a \ar@{-}[rr] \ar@{.}[dr] && x_b \ar@{.}[dl] \\ & \overline x_1 } \]
If $|S| = 3$, then $S$ must be a triangle of the form $\{ x_1, x_a, x_b \}$ where $x_1$ is true. Now one of $x_a$ and $x_b$ is made worse off in this triangle compared to the triangle $\{ x_1, x_a, x_b \}$ which is part of $\pi$ (by monotonicity on triangles), in this case $x_b$. Thus $x_b$ does not deviate. Since $\pi(\overline x_1) = \pi(x_b)$ but $\pi'(\overline x_1) \neq \pi'(x_b)$, $\overline x_1$ must deviate. But $\overline x_1$ is now strictly worse off by triangle-appreciation and monotonicity on triangles, a contradiction.

	\[ \xymatrix@R-1pc@C-1pc{ & i \ar@{-}[dr] \ar@{-}[dl] & \\ j \ar@{-}[rr] \ar@{.}[dr]|? && k \ar@{.}[dl]|? \\ & \ell } \]
Suppose $|S| = 4$, with $S = \{i,j,k,\ell\}$ with $\{i,j,k\}$ forming a triangle, and $i$ and $\ell$ being enemies. For this case we only have to worry about SIS-deviations. Suppose first that $\ell$ is friends with both $j$ and $k$. Now in $\pi$, either $\{i,j,k\}$ formed a coalition, in which case $S$ is worse for $i$ by intolerance in triangles, or $\{j,k,\ell\}$ formed a coalition. In the latter case, $\ell$ is made worse off by $i$ joining (using intolerance in triangles), so that the deviation is not SIS.

Hence $\ell$ is not friends with both $j$ and $k$. Let's condition on whether $\{i,j,k\}\in \pi$ or not. If $\{i,j,k\}\in \pi$, then $i$ is worse off in $S$ by intolerance in triangles. So suppose $\{i,j,k\}\not\in \pi$. If $\ell$ has no friends in $S$, then 1 of $i$, $j$, $k$ is member of a triangle in $\pi$ that is better than $\{i,j,k\}$, and thus is worse off in $S$ by intolerance in triangles, so no SIS. Similarly if $\ell$ has exactly 1 friend in $S$, say $j$ (it cannot be $i$ who has exactly 2 friends), then $j$ must be a literal player, and one of $i$ and $k$ was better off in $\pi$ (where they were in a better triangle than $\{i,j,k\}$, and without enemies), so this is not an SIS-deviation. \qed

\begin{lemma}
	No deviator has $3$ or $4$ friends in $\pi'$.
\end{lemma}

\begin{proof}
	Suppose $i\in S \in \pi$ is a deviator where $S$ includes 3 or 4 friends. We condition on the type of $i$.
	
	\textit{$i$ is a clause player:} Say $i$ has name $c_i$. $c_i$'s 3 friends are pairwise enemies, so at least 2 of them deviate, including another clause player $j\in\{c_{i+1}, c_{i-1}\}$. If $j$ is matched then $j$ is worse off since there are 2 enemies in $S$ and $j$ is of degree 2, contradiction. So $j$ is single, and hence $j = c_{i-1}$. Then the other clause friend $c_{i+1}$ of $c_i$ is matched and now worse off so not an SIS deviation. With $\llangle$strict$\rrangle$ \toxicc11-toxicity, $c_{i-1}$ is not made better off unless $c_{i-1}$ has 2 friends ($c_i$ and $c_{i-2}$) in $S$. Since $c_{i+1}$ does not deviate, $c_{i-2}$ must deviate. But $c_{i-2}$ is matched and has enemies $c_i$ and $c_{i+1}$, so must have 3 friends which is impossible since $c_{i-2}$ has degree 2, contradiction.
	
	\textit{$i$ is a player of type $x_a$, $x_b$, $x_c$, or $x_d$:} Say $x_a$, and $\pi(x_a) = \{ x_a, x_b, \overline x_1 \}$. By assumption, $x_a$ has 3 friends in $S$, so both $x_1$ and $\overline x_1$ are in $S$.
	
	\begin{itemize}
		\item By intolerance in triangles, this cannot be an SIS deviation unless $\overline x_1$ has at least 3 friends in $S$, so there is an extra friend $j\in S$:
		\[ \xymatrix@R-1pc@C-1pc{ & x_1 \ar@{-}[dr] \ar@{-}[dl] & \\ x_a \ar@{-}[rr] \ar@{-}[dr] && x_b \ar@{-}[dl] \\ & \overline x_1 \ar@{-}[d] \\ & j } \]
		\item $\langle$By weak \toxicc21-toxicity, this cannot be an SNS or SSNS deviation unless 1 of $x_1$ or $\overline x_1$ has at least 3 friends in $S$ (since one of them deviates). Say $\overline x_1$ has extra friend $j\in S$, also giving us the picture above.$\rangle$
	\end{itemize} 
	
	Now both $x_1$ and $j$ have at least 2 enemies in $S$, and one of them must deviate, hence have at least 3 friends in $S$. If this is $j$, then $|S| \ge 7$ and $x_a$ is unhappier by weak \toxicc33-toxicity, contradiction. So $j$ doesn't deviate and $x_1$ has an extra friend in $S$, say $k$. Since $j$ doesn't deviate, $k$ does deviate. Since $x_1$ has no lonely friends, $k$ is matched. Hence $k$ has 3 friends in $S$. Then $|S| \ge 7$ and $x_a$ is unhappier by weak \toxicc33-toxicity, contradiction. 
	
	\textit{$i$ is a literal player, say $x_1$}: Suppose first that $S$ includes $\overline x_2$ and $c(x_1)$, so that there are 3 friends of $x_1$ in $S$ that are pairwise enemy, and hence 2 friends who must be deviating. Of those 2, at most 1 is of the $x_a$ or $x_b$ kind; such a player must have 3 friends in $S$ by intolerance in triangles. All other types must have 3 friends in $S$ by toxicity. The 2 deviating friends of $x_1$ thus together contribute at least 3 extra friends to $S$, and hence $|S| \ge 7$. Then at least 1 of the deviating friends of $x_1$ has 3 friends in $S$ but also 3 enemies, and is thus worse off in $S$, a contradiction.
	
	Suppose otherwise that $S$ includes exactly 3 friends of $x_1$, including $x_a$ and $x_b$, and also $j\in \{\overline x_2, x(c_1)\}$. Suppose $x_a$ and $x_b$ do not deviate. Then everyone else (except possibly $\overline x_1$) in $S$ must deviate; in particular $j$ must deviate and hence have 3 friends in $S$ (which means 2 extra agents in $S$), who each must also deviate, which means 1 extra agent for each of the 2 extra agents in $S$. Hence $|S| \ge 7$ and $x_1$ is worse off in $S$ by weak \toxicc33-toxicity, a contradiction. Otherwise, at least 1 of $x_a$ and $x_b$ is deviating. Now unless $\overline x_1\in S$, by intolerance in triangles one of $x_a$ or $x_b$ is worse off, preventing this from being an SIS deviation. Otherwise by weak \toxicc21-toxicity, the deviator from $x_a$ and $x_b$ needs 3 friends. Either way we conclude $\overline x_1 \in S$. Now at least 1 of $\overline x_1$ and $j$ must be deviating, and thus must have 3 friends in $S$. If $j$ is deviating then $|S|\ge 7$. If $\overline x_1$ is deviating, then either its extra friend or $j$ is deviating, bringing $|S|\ge 7$, so that $x_1$ is worse off by weak \toxicc33-toxicity.
\end{proof}

Now suppose $i$ is a matched deviator. Then by the lemmas, $i$ must end up in a coalition $S\in\pi'$ which includes exactly 1 friend of $i$. By weak $\toxicc11$-toxicity, $S$ must be a pair, $S = \{i,k\}$, say.

We now go through each type $i$ of matched player and show that $i$ does not deviate into a pair, and is thus not a deviator.

\begin{itemize}
	\item By triangle-appreciation, $x_a$, $x_b$, $x_c$, $x_d$, and false variable occurrences are strictly better off than in any pair, so none of them are deviators.
	\item True variable occurrences: all preferred players are in triangles consisting of non-deviators.
	\item $c_i$ players don't deviate for the same reason as before.
\end{itemize}

Hence (together with the lemmas above) no matched player is a deviator.

Now consider a lonely player $i$. If $i$ deviates, then $i$ ends up in a coalition $S\in\pi'$ consisting of lonely players and possibly 2 players that are in a pair in $\pi$ or 3 players in a triangle in $\pi$ (because no matched players are deviators). If $S$ consists entirely of lonely players, then by $\llangle$strict$\rrangle$ \toxicc01-toxicity, $i$ is $\llangle$worse off$\rrangle$ not better off, so won't deviate (SCR: no-one is strictly better off). Hence $S$ also contains a pair or triangle. At least 1 of these 2 or 3 players is enemies with $i$ and is thus worse off by weak \toxicc11-toxicity or intolerance in triangles (since all lonely players are enemies to triangles), so this is not an SIS-deviation. $\langle$With $\llangle$strict$\rrangle$ \toxicc11-toxicity, $i$ is not better off in $S$, so not an SNS-deviation $\llangle$SSNS-deviation$\rrangle$.$\rangle$ 

Thus we conclude that no player is a deviator. Hence $\pi$ is SIS-stable $\langle$SNS-stable$\rangle$ $\llangle$SSNS-stable$\rrangle$.

\subsubsection{Stable $\Rightarrow$ Satisfiable.}

Suppose $\pi$ is a core-stable partition of the game. We show that $\varphi$ is then satisfiable. The first 3 lemmas are proved exactly as before.

\begin{lemma}
	The $9$ players of a clause cannot all be together in the same coalition in $\pi$.
\end{lemma}

\begin{lemma}
	For any given clause, at least $1$ of its players must be together with their literal player in $\pi$.
\end{lemma}

\begin{lemma}
	No $c_i$ can have $3$ friends in $\pi$. A $c_i$ together with its literal has at most $1$ enemy in $\pi$.
\end{lemma}

\begin{lemma}
	$x_1$ and $\overline x_1$ cannot both be together with their clause player.
\end{lemma}

\begin{proof}
	Suppose they are. Note that if either $x_1$ or $\overline x_1$ was together with either of $x_a$ or $x_b$ (cannot be together with both) then $\{x_a, x_b\}$ blocks. But then by triangle-appreciating and monotonicity $\{ x_1, x_a, x_b \}$ blocks.
\end{proof}

\begin{lemma}
	$x_1$ and $\overline x_2$ cannot both be together with their clause player.
\end{lemma}

\begin{proof}
	Suppose they are. Since the clause players can have at most 1 enemy, $x_1$ and $\overline x_2$ cannot be in the same coalition in $\pi$. Also $x_1$ cannot be together with either $x_a$ or $x_b$ since then $\{x_a, x_b\}$ blocks, and similarly $\overline x_2$ cannot be together with either $x_c$ or $x_d$ since then $\{x_c, x_d\}$ blocks. Hence $\{ x_1, \overline x_2 \}$ blocks by consistency on pairs.
\end{proof}

Define a propositional assignment $\mathcal A$ that sets literals that are in a coalition with their clause player true. By the last two lemmas, this is well-defined. By the lemma before, each clause has at least 1 literal that is set true by $\mathcal A$. Hence $\mathcal A$ satisfies $\varphi$.

\section{Class Properties}

In this section, we check in more detail that the conditions of our theorems are satisfied for various classes. All of these are routine.

Throughout we may assume $N\neq\emptyset$ and so $n\ge 1$, because all conditions are satisfied vacuously by the empty hedonic game. It is useful to note that strict $\tox31$-toxicity implies intolerance in triangles, so that we do not need to check this condition in all cases.
\begin{itemize}
\item \textbf{IRCL}

Given $N$ and $\orderings$ with friend sets $(F_i)_{i\in N}$, run the following algorithm producing IRCL lists.
\begin{enumerate}
	\item List $\{ (j,k) \in F_i\times F_i : j >_i k$ or $j\sim_i k$ and $j$ comes earlier than $k$ in listing of $N\}$.
	\item Sort this list according to $(j,k) \gg (j', k')$ iff $j\ge_i j'$ and $k >_i k'$. Break ties arbitrarily.
	\item Output this list with entries written as triangles $\{i,j,k\}$, any two entries separated by $\succ_i$.
	\item\label{step:irclpairs} Output $F_i\cup\{i\}$, written as pairs $\{ i,j \}$, with $>_i$ replaced by $\succ_i$ and $\sim_i$ replaced by $\sim_i$.
	\item End of output.
	\end{enumerate}
	Clearly this algorithm terminates in polynomial time. The game described by the output is triangle-appreciating in all senses, because all friend-triangles come before all other coalitions. By step~\ref{step:irclpairs}, the game is consistent on pairs. Because no coalition including an enemy is listed, they are not individually rational, so strict $\tox{k}{1}$-toxicity is satisfied.
	
\item	\textbf{Stable Roommates}
	
	This is the game produced by the IRCL-algorithm when we start it in step~\ref{step:irclpairs}. So it is consistent on pairs, and strict $\tox{k}{1}$-toxicity is satisfied.
	
	Neither of the triangle conditions is satisfied.
	
\item 	\textbf{$\mathcal W$-Games}
	
	\textit{Consistency on pairs.} For agents $j,k\in F_i\cup\{i\}$, by definition $\W_i(\{i,j\}) = j$ [for both cases $j\neq i$ and $j=i$], and so $\{i,j\} \pref_i \{i,k\}$ iff $\W_i(\{ i,j \}) \ge_i \W_i(\{i,k\})$ iff $j\ge_i k$.
	
	\textit{Strict $\tox k1$-toxicity.} If $S$ includes an enemy $e\in E_i$, then $\W_i(S) \le_i e <_i i = \W_i(\{i\})$, so $S \prec_i \{i\}$.
	
	We do not have triangle-appreciating. We do have monotonicity on triangles, but this is irrelevant.
	
\item	\textbf{$\mathcal{WB}$-Games}
	
	The analysis is very similar to the case of $\W$-games.
	
\item	\textbf{Additively Separable Games}
	
	Use $\utilties{-(n^2+2n),4}$-utilities.
	
	\textit{Consistency on pairs.} $\{i,j\} \pref_i \{i,k\}$ iff $v_i(j) \ge_i v_i(k)$ iff $j \ge_i k$.
	
	\textit{Strict $\tox k1$-toxicity.} Suppose $S$ contains an enemy $e\in E_i$. Then $v_i(S) = \sum_{j\in S} v_i(j) \le (n-2) \times (n+4) - (n^2 + 2n) = -8 < 0$, so $S \prec_i \{i\}$.
	
	\textit{Triangle-appreciating.}
	If $j,k\in F_i$ distinct with $j\ge_i k$ then $v_i(\{i,j,k\}) = v_i(j) + v_i(k) > v_i(j) = v_i(\{ i,j \})$ since $v_i(j) \ge 4 > 0$ in $\utilties{-(n^2+2n),4}$-utilities.
	
	\textit{Monotone on triangles.} Suppose $j,j',k,k'\in F_i$ are such that $j \ge_i j' >_i k >_i k'$. Then $v_i(\{i,j,k\}) = v_i(j) + v_i(k) > v_i(j') + v_i(k') = v_i(\{i,j',k'\})$.
	
\item	\textbf{Hedonic Coalition Nets}
	
	We essentially encode the additively separable game with $\utilties{-(n^2+2n),4}$-utilities from above as a hedonic coalition net. Write $E_i = \{ e^i_1,\dots, e^i_k \}$. Then use the net
	\begin{align*}
	j &\longmapsto_i v_i(j) \qquad \text{for friends $j\in F_i$,} \\
	e^i_1\lor \dots\lor e^i_k &\longmapsto_i -(n^2 + 2n)
	\end{align*}
	As noted in the paper, $|F_i| \le 4$ for all $i$, so the net above uses at most 5 formulas per agent. We can verify the properties exactly as we did for the Additively Separable Game above; toxicity goes through since in our check we only used the presence of a single enemy.
	
	All weights we used in the net are of size polynomial in $n$. Since we always have an explicit list of $N$ as input to our algorithms, we have $n$ available in unary, so we are allowed to write the weights in unary.
	
\item	\textbf{Fractional Hedonic Games}
	
	Use $\utilties{-(n^2+5n),7}$-utilities.
	
	\textit{Consistent on Pairs.} For $j,k\in F_i\cup\{i\}$, $\{i,j\} \pref_i \{i,k\}$ iff $v_i(j)/2 \ge v_i(k)/2$ iff $j \ge_i k$.
	
	\textit{Strictly $\tox k1$-toxic.} Suppose $S$ contains an enemy $e\in E_i$. Then $v_i(S) = 1/|S| \sum_{j\in S} v_i(j) \le \sum_{j\in S} v_i(j) \le (n-2) \times (n+7) - (n^2+5n) = -14 < 0$.
	
	\textit{Triangle appreciating.} Let $j, k\in F_i$ be distinct with $j\ge_i k$ and satisfying the closeness condition, which implies $v_i(j) - v_i(k) \le 2$. Then $v_i(\{i,j,k\}) = (v_i(j) + v_i(k))/3 \ge (2v_i(j) -2)/3 = \frac23 v_i(j) - \frac23 > v_i(j)/2 = v_i(\{i,j\})$ because $v_i(j) \ge 5$ by choice of utilities.
	
	\textit{Monotone on triangles.} Suppose $j,j',k,k'\in F_i$ are such that $j \ge_i j' >_i k >_i k'$. Then $v_i(\{i,j,k\}) = (v_i(j) + v_i(k))/3 > (v_i(j') + v_i(k'))/3 = v_i(\{i,j',k'\})$.
	
\item	\textbf{Social FHGs.}
	
	Use $\utilties{0,7n}$-utilities.
	
	Note that because $7n \ge 7$, all the `positive properties' hold as they did for straight FHGs. We only need to check the negative properties.
	
	\textit{$\tox01$-toxic.} A coalition $S$ in which $i$ only has enemies obtains value 0.
	
	\textit{Weakly $\tox11$-toxic.} Let $S$ be $\tox11$. Note that then $|S| \ge 3$. So $v_i(S) = \frac1{|S|}\sum v_i(j) \le \frac13 8n < \frac12 7n$, where $\frac12 7n$ is the minimal utility obtained in pairs.
	
	\textit{Weakly $\tox22$-toxic.} Follows from $\frac25 8n < \frac12 7n$.
	
	\textit{Weakly $\tox33$-toxic.} Follows from $\frac37 8n < \frac12 7n$.
	
	\textit{Intolerant in triangles.} Follows because $(a+b)/3 > (a+b)/n$ for all $n>3$.
	
\item	\textbf{Median Games.}
	
	Use $\utilties{0, 5}$-utilities.
	
	\textit{$\tox01$-toxic.} Let $S$ be $\tox01$. Then all utilities in $S$ are 0, so it is evaluated at 0, so $\{i\} \pref_i S$.
	
	\textit{Weakly $\tox kk$-toxic.} Let $S$ be $\tox kk$. Then the median value corresponds to a value or average of $v_i(i)=0$ or $v_i(\text{enemy}) = 0$, so is 0. On the other hand, in a pair with a friend, $i$ obtains utility at least $5/2 > 0$ by choice of utilities. Hence $\{i,j\} \succ_i S$ for friends $j$.
	
	\textit{Triangle appreciating.} Let $j, k\in F_i$ be distinct with $j\ge_i k$ and satisfying the closeness condition, which implies $v_i(j) - v_i(k) \le 2$. Then $v_i(\{i,j,k\}) = v_i(k) \ge v_i(j) - 2 > v_i(j)/2 = v_i(\{i,j\})$ because $v_i(j) \ge 5$ by choice of utilities.
	
	\textit{Monotone on triangles.} Suppose $j,j',k,k'\in F_i$ are such that $j \ge_i j' >_i k >_i k'$. The median value of $\{i,j,k\}$ is $k$, the median value of $\{i,j',k'\}$ is $k'$. Hence $\{i,j,k\}\succ_i \{i,j',k'\}$.
	
	\textit{Intolerance in triangles.} Easy to see under the assumption that preferences are strict.
	
\item	\textbf{Midrange.}
	
	Use $\utilties{-3n,1}$-utilities.
	
	\textit{Consistent on pairs.} The midrange of a pair $\{i,j\}$ is (depending on some choices) $v_i(j)$ or $v_i(j)/2$.
	
	\textit{Strict $\tox k1$-toxicity.} Suppose $S$ contains an enemy $e\in E_i$. Then $v_i(S) = \frac12\mathcal B + \frac12\W = \frac12\mathcal B - \frac12 3n \le \frac12 (n+1 -3n) = -n + \frac12 < 0$ since $n\ge 1$.
	
\item	\textbf{$\ell$-Approval.}
	
	Assume $\ell\ge 4$, and use $\llbracket -6\ell n,4\rrbracket$-utilities.
	
	\textit{Consistent on pairs.} A pair is valued with the utility of the partner.
	
	\textit{Strictly $\tox k1$-toxic for $k<\ell$.} Suppose $S$ is $\tox k1$ with $k < \ell$. Then $v_i(S) = Y_1 + \dots + Y_k + Y_{k+1} = Y_1 + \dots + Y_k -6\ell n \le k(n+4) - 6\ell n < \ell(n+4) - 6\ell n = \ell (-5n + 4) < 0$ since $\ell\ge 0$ and $n\ge 1$.
	
	The checks of the remaining conditions are identical to the case of additively separable games, noting that there we never sum more than the first 3 entries.
\end{itemize}

\end{document}